\documentclass[11pt]{article}
\usepackage[noadjust]{cite}
\usepackage{tikz}
\usepackage{longtable}
\usepackage{caption}
\usepackage{multirow,multicol}
\usepackage{epsf, graphicx}
\usepackage{latexsym,amsfonts,amsbsy,amssymb}
\usepackage{amsmath,amsthm}
\usepackage{cite}
\usepackage{cleveref}
\usepackage{indentfirst}
\usepackage{bm}
\usepackage [latin1]{inputenc}
\usepackage{enumitem}
\usepackage{lipsum}
\usepackage{tabularx}
\usepackage{adjustbox}
\usepackage{diagbox}
\usetikzlibrary{tikzmark, arrows.meta}
\setenumerate[1]{itemsep=0pt,partopsep=0pt,parsep=\parskip,topsep=5pt}
\setitemize[1]{itemsep=0pt,partopsep=0pt,parsep=\parskip,topsep=5pt}
\setdescription{itemsep=0pt,partopsep=0pt,parsep=\parskip,topsep=5pt}
\allowdisplaybreaks[4]
\makeatletter

\@addtoreset{equation}{section} \makeatother
\newtheorem{Theorem}{Theorem}[section]
\newtheorem{Lemma}{Lemma}[section]
\newtheorem{Corollary}{Corollary}[section]
\newtheorem{Remark}{Remark}[section]
\newtheorem{Definition}{Definition}[section]

\newtheorem{Example}{Example}[section]
\makeatletter \@addtoreset{equation}{section} \makeatother
\begin{document}
	
	\title{Function-Correcting RT Codes: A Coding Framework for Parallel Channels  \let\thefootnote\relax\footnotetext{E-Mail addresses: huiyingliu@mails.ccnu.edu.cn (Huiying Liu),  hwliu@ccnu.edu.cn (Hongwei Liu), smesnager@univ-paris8.fr (Sihem Mesnager).}}
	\author{  Huiying Liu\textsuperscript{1}, Hongwei Liu\textsuperscript{1}, Sihem Mesnager\textsuperscript{2} }
	\date{\small 1. School of Mathematics and Statistics, Central China Normal University, Wuhan, 430079, China\\
    2. LAGA, Department of Mathematics, Universities of Paris VIII and Paris XIII, CNRS, UMR 7539 and Telecom ParisTech, 91120 Palaiseau, France}
	\maketitle
	%MS+++++++++++++++++++++ Abstract +++++++++++++++++++++++++

	{\noindent\small{\bf Abstract:}
Function-correcting codes have recently emerged as an effective approach to reducing coding redundancy when the objective is to reliably recover only the prescribed function values of the transmitted information, rather than the complete message itself. As a natural generalization of the Hamming metric, the Rosenbloom--Tsfasman metric (RT-metric) has attracted considerable attention due to its relevance to communication over parallel channels. Motivated by these two research directions, we introduce in this paper \emph{function-correcting RT codes} (FCRTCs), which extend the framework of function-correcting codes from the Hamming metric to the more general RT metric.

To investigate the fundamental problem of optimal redundancy, we introduce a new class of matrix codes, called \emph{irregular RT-distance codes}, and establish a direct correspondence between the redundancy optimization problem for FCRTCs and the construction of such codes. This connection provides a unified framework for deriving redundancy bounds under the RT metric.

For arbitrary functions, we establish general upper and lower bounds on the optimal redundancy of FCRTCs. We then specialize our analysis to several important classes of functions, namely RT-weight functions, RT-weight distribution functions, and RT-locally-two-valued binary functions, for which substantially sharper bounds are obtained by exploiting their specific structural properties. Furthermore, we present explicit constructions of FCRTCs for RT-weight distribution functions and prove that these constructions attain the optimal redundancy bounds in several cases.

The proposed framework considerably extends the theory of function-correcting codes to the RT metric, provides new theoretical tools for studying reliable communication over parallel channels, and opens new perspectives for redundancy-efficient coding beyond the classical Hamming setting.
}
	
	\vspace{1ex}
	{\noindent\small{\bf Keywords:}
		Rosenbloom-Tsfasman metric; Function-correcting codes; Irregular RT-distance codes; Optimal redundancy.}

	2020 \emph{Mathematics Subject Classification}: 
    94B60, 94B65
	\section{Introduction}

In modern communication systems, information is encoded into codewords before being transmitted through a noisy channel. To ensure reliable communication, redundancy is intentionally introduced during the encoding process so that transmission errors can be detected and corrected at the receiver. In the classical framework of coding theory, the decoding procedure aims to recover the entire transmitted message, thereby enabling any desired information derived from it to be computed.

However, this requirement is unnecessarily restrictive in numerous practical scenarios. Indeed, many applications require only the value of a prescribed function of the transmitted message rather than the message itself. Recovering the complete information therefore introduces additional redundancy that is not always essential. Motivated by this observation, Lenz, Bitar, Wachter-Zeh and Yaakobi~\cite{lbwy} introduced the notion of \emph{function-correcting codes} in 2023, thereby establishing a new coding paradigm that explicitly exploits the structure of the underlying function.

In the function-correcting framework, both the sender and the receiver are assumed to know the target function in advance. Consequently, for a prescribed error-correction capability and a given function, the minimum-distance requirement is imposed only on codewords corresponding to messages having distinct function values. This represents a fundamental departure from classical error-correcting codes, where every pair of distinct codewords must satisfy the prescribed distance constraint. By relaxing the distance requirements when function values coincide, function-correcting codes can significantly reduce coding redundancy while still guaranteeing the correct recovery of the desired function value. Clearly, every classical error-correcting code is automatically a function-correcting code, since recovering the entire transmitted message immediately determines the corresponding function value.

In their pioneering work~\cite{lbwy}, Lenz \emph{et al.} established fundamental bounds on the optimal redundancy of function-correcting codes and presented explicit constructions for Hamming-weight and Hamming-weight-distribution functions. Since then, this research direction has developed rapidly and has been extended to a variety of channel models and distance metrics.

For symbol-pair read channels, Xia, Liu and Chen~\cite{xlc} investigated function-correcting codes with the symbol-pair metric and derived several bounds on their optimal redundancy. Subsequently, Singh, Singh and Yaakobi~\cite{ssy} generalized these results to the more general $b$-symbol read channel model. Rajan~\cite{sr} further established a Plotkin-type bound for function-correcting codes under the $b$-symbol metric. In the context of the Lee metric, Hareesh, Rashid Ummer and Sundar Rajan~\cite{hrs} proposed explicit constructions of function-correcting codes, while Verma and Singh~\cite{vs} investigated corresponding bounds on the optimal redundancy. Extending the framework beyond finite fields, Liu and Liu~\cite{ll} generalized function-correcting codes to finite chain rings equipped with the homogeneous metric. More recently, Singh and Singh~\cite{ss} initiated the study of function-correcting codes for insertion--deletion channels, further illustrating the broad applicability and flexibility of this emerging coding paradigm.

 In 2025, Zhang, Xu, Zhang and Ge~\cite{zxzg} employed Gray codes to improve several results established in~\cite{lbwy,xlc}. Another active research direction concerns the design of function-correcting codes tailored to specific classes of functions. In this context, Premlal and Sundar Rajan~\cite{pr} investigated function-correcting codes for linear functions. Further developments for various families of functions can be found in~\cite{rrhh,dmkpr,pbmr}. Verma, Singh and Singh~\cite{vss} studied function-correcting $b$-symbol codes associated with locally $(\lambda,\rho,b)$-functions, while Verma and Singh~\cite{vsh} investigated function-correcting codes endowed with the homogeneous metric for linear and locally bounded functions. Additional advances along these directions are reported in~\cite{rml,rrhh data protection,rrhh partition codes}. Collectively, these contributions demonstrate both the versatility of the function-correcting coding framework and its increasing importance across a broad spectrum of communication models and application-driven settings.

Parallel to these developments, the Rosenbloom--Tsfasman metric (RT-metric) has become one of the fundamental metrics for modelling communication over parallel channels. Originally introduced by Niederreiter~\cite{n} in 1987 in the study of point sets, and later systematically developed by Rosenbloom and Tsfasman~\cite{rt}, the RT-metric naturally extends the classical Hamming metric while capturing error structures arising in parallel transmission systems. Rosenbloom and Tsfasman~\cite{rt} established a Singleton-type bound for codes equipped with the RT metric, thereby laying the theoretical foundations of this metric.

Codes endowed with the RT metric are naturally represented as matrix codes. In particular, when each codeword consists of a single row, the RT metric coincides with the classical Hamming metric, showing that the latter appears as a special case of the former. Consequently, many classical coding-theoretic problems have subsequently been revisited within the RT-metric framework.

Over the past decades, a rich theory of RT-metric codes has been developed. Dougherty and Skriganov~\cite{ds} established a MacWilliams-type identity for RT-metric codes, while Skriganov~\cite{s} investigated MDS codes together with the properties of their dual codes. Quistorff~\cite{q} derived new bounds for RT-metric codes and improved the RT analogue of the Plotkin bound. Zhou, Lin and Abdel-Ghaffar~\cite{zla} constructed BCH codes endowed with the RT metric, and subsequently employed Galois--Fourier transforms associated with Hasse derivatives to investigate structural properties of Reed--Solomon and BCH codes under this metric. More recently, Xu, Xu and Du~\cite{xxd} introduced linear complementary dual codes with the RT metric, while Wang, Li and Heng~\cite{wlh} established new upper bounds on the minimum RT distance of linear matrix codes and proposed corresponding constructions. Further developments on RT-metric codes can be found in~\cite{dsm,hk,qgp,j,msg}.

Although numerous families of classical codes have successfully been generalized from the Hamming metric to the RT metric, the function-correcting paradigm has not yet been investigated in this setting. This observation naturally raises the following question:

\begin{center}
\emph{Can the theory of function-correcting codes be extended to the Rosenbloom--Tsfasman metric while preserving its redundancy advantages?}
\end{center}

The primary objective of this paper is to answer this question affirmatively. To this end, we introduce \emph{function-correcting RT codes} (FCRTCs), which constitute the RT-metric counterpart of function-correcting codes under the Hamming metric. Since the RT metric is defined on matrix spaces, FCRTCs naturally take the form of systematic matrix codes. Each codeword is decomposed into two components: an information matrix that carries the source data and a redundancy matrix that provides protection against transmission errors. Within this framework, our main objective is to investigate the fundamental problem of determining the optimal redundancy required by FCRTCs and to establish general bounds, along with explicit constructions, for several important classes of functions.

To address the redundancy optimization problem, we introduce a new class of matrix codes endowed with the RT metric, referred to as \emph{irregular RT-distance codes}. In these codes, the pairwise RT-distance constraints between codewords are prescribed by a nonnegative matrix $D$; accordingly, we also refer to them as \emph{$D$-RT codes}. In addition, we introduce two new mathematical objects, namely \emph{RT-distance matrices} and \emph{RT-function-distance matrices}, which provide an appropriate framework for analysing function-correcting RT codes.

These notions establish a direct connection between the optimal redundancy problem for FCRTCs and the construction of irregular RT-distance codes. This correspondence is one of the central ideas of the paper, as it enables the study of FCRTCs by analysing a more general class of matrix codes. In particular, for an arbitrary symmetric matrix $D$ and assuming the number of rows of the codewords is fixed, we derive both upper and lower bounds on the minimum number of columns required for $D$-RT codes.

Building upon this general framework, we first establish upper and lower bounds on the optimal redundancy of FCRTCs for an arbitrary function, thereby extending several fundamental results in~\cite {lbwy} from the Hamming metric to the RT metric. We then focus on several important classes of functions for which significantly stronger results can be obtained by exploiting their specific algebraic and combinatorial properties.

For RT-weight functions, we derive sharper bounds on the optimal redundancy of FCRTCs by leveraging the structural properties of the RT weight, thereby extending Lemma~6 of~\cite{lbwy}. For RT-weight distribution functions, we propose explicit constructions of FCRTCs and prove that these constructions attain the optimal redundancy in several parameter regimes, thus generalizing Lemma~8 of~\cite{lbwy}. Furthermore, for RT-locally-two-valued binary functions, we determine the optimal redundancy in several cases, extending the corresponding result given in Lemma~5 of~\cite{lbwy}.

The remainder of this paper is organized as follows. Section~2 introduces the notation and recalls the basic definitions required throughout the paper. In Section~3, we develop the general theoretical framework by establishing upper and lower bounds on the optimal redundancy of FCRTCs for arbitrary functions. This section also investigates irregular RT-distance codes and derives general bounds on the minimum column length of $D$-RT codes. Sections~4, 5, and~6 are devoted to RT-weight functions, RT-weight distribution functions, and RT-locally-two-valued binary functions, respectively, in which the general theory is refined to obtain stronger results and explicit constructions. Finally, Section~7 concludes the paper and discusses several directions for future research.

\section{Preliminaries}

Throughout this paper, let $S$ be a set and denote its cardinality by $|S|$. Let
$k$, $r$, $h$ and $t$ be positive integers. We denote by
$\mathbb{F}_2^{h\times k}$ the vector space of all $h\times k$ matrices over the binary
field $\mathbb{F}_2$.

This section introduces the basic concepts and notation used throughout the paper. We first recall the notions of the Rosenbloom-Tsfasman weight and distance on matrix spaces, which form the metric foundation of our work. We then introduce function-correcting RT codes together with a new family of matrix codes, called irregular RT-distance codes, that will play a central role in the development of our theoretical framework

\subsection{Rosenbloom--Tsfasman distance}

We begin by recalling the definitions of the Rosenbloom--Tsfasman (RT) weight and the associated RT distance over the matrix space $\mathbb{F}_2^{h\times k}$. These notions generalize the classical Hamming weight and Hamming distance and provide the natural metric framework for coding over parallel channels.

\begin{Definition}%{\rm (\cite{rt})}
\label{Def1 RT-weight}

Let
\[
A=(\boldsymbol{a}_1,\boldsymbol{a}_2,\ldots,\boldsymbol{a}_k)
\in\mathbb{F}_2^{h\times k},
\]
where
\[
\boldsymbol{a}_i=(a_{1i},a_{2i},\ldots,a_{hi})^\intercal
\in\mathbb{F}_2^{h\times1},
\qquad
i\in\{1,2,\ldots,k\}.
\]

The \emph{Rosenbloom--Tsfasman weight} (or \emph{RT-weight}) of the column vector
$\boldsymbol{a}_i$ is defined by
\begin{equation}\label{Formula RT-weight}
           w_{RT}(\boldsymbol{a}_i)=\left\{
  	         \begin{aligned}
  	          &h-\min\{\,s-1\,|\,a_{si}\neq 0\},&~\text{if } \boldsymbol{a}_i\neq 0,\\
                &0,&~\text{if } \boldsymbol{a}_i=0.
               \end{aligned}
               \right.
        \end{equation}
        
        The \emph{RT distance} between two column vectors
$\boldsymbol{a}_i$ and $\boldsymbol{a}_j$ is defined by
\[
d_{RT}(\boldsymbol{a}_i,\boldsymbol{a}_j)
=
w_{RT}(\boldsymbol{a}_i-\boldsymbol{a}_j).
\]
Furthermore, the RT-weight of the matrix $A$ is given by
\[
w_{RT}(A)=\sum_{i=1}^{k}w_{RT}(\boldsymbol{a}_i).
\]

    \end{Definition}

For illustration, consider the matrix
\[
A=
\begin{pmatrix}
1 & 0 & 1\\
0 & 1 & 0
\end{pmatrix}
\in\mathbb{F}_2^{2\times3}.
\]
Its RT-weight is
\[
w_{RT}(A)=5.
\]
On the other hand, if the entries of $A$ are viewed as a vector in $\mathbb{F}_2^{6}$, its Hamming weight is equal to $3$. This example clearly illustrates that the RT-weight and the Hamming weight measure different structural properties of the same object.

It is worth emphasizing that the RT-weight naturally extends the classical Hamming weight. Indeed, when $h=1$ in Definition~\ref{Def1 RT-weight}, every codeword consists of a single row, and the RT-weight of any matrix
$A\in\mathbb{F}_2^{1\times k}$ coincides exactly with its Hamming weight. Consequently, the RT metric may be viewed as a genuine generalization of the Hamming metric from vector spaces to matrix spaces, making it particularly well suited for coding problems over parallel channels.

\begin{Definition}\label{def RT-distance}
Let $\mathbb{F}_2^{h\times k}$ denote the vector space of all $h\times k$ matrices over $\mathbb{F}_2$. The \emph{Rosenbloom--Tsfasman distance} (abbreviated as the \emph{RT distance}) between two matrices
$A,B\in\mathbb{F}_2^{h\times k}$ is defined by
\[
d_{RT}(A,B)=w_{RT}(A-B).
\]
\end{Definition}
    
    To illustrate Definition~\ref{def RT-distance}, consider the matrices
\[
A=
\begin{pmatrix}
1&0&1\\
0&1&0
\end{pmatrix},
\qquad
B=
\begin{pmatrix}
1&1&1\\
0&1&0
\end{pmatrix},
\qquad
C=
\begin{pmatrix}
0&1&0\\
1&0&1
\end{pmatrix}
\in\mathbb{F}_2^{2\times3}.
\]
A straightforward computation yields
\[
d_{RT}(A,B)=2,\qquad
d_{RT}(A,C)=6,\qquad
d_{RT}(B,C)=4.
\]

It is well known that $d_{RT}$ satisfies the axioms of a metric on
$\mathbb{F}_2^{h\times k}$.
Consequently, the pair
$\left(\mathbb{F}_2^{h\times k},d_{RT}\right)$
forms a metric space, which provides the natural mathematical framework for studying coding problems over parallel channels.

\subsection{Function-correcting RT codes and irregular RT-distance codes}

Having introduced the RT metric, we now define the main coding framework that we investigate in this paper. We first introduce function-correcting RT codes (FCRTCs), which extend the notion of function-correcting codes to the Rosenbloom--Tsfasman metric. We then introduce a more general family of matrix codes, called \emph{irregular RT-distance codes}, along with two auxiliary matrix representations: RT-distance matrices and RT-function-distance matrices. These new notions will provide the key tools for transforming the redundancy optimization problem of FCRTCs into an equivalent problem involving irregular RT-distance codes.

The objective of function-correcting RT codes is to guarantee the reliable recovery of function values after transmission through parallel channels, rather than the recovery of the entire transmitted information matrix. As in the classical framework of function-correcting codes, the code construction depends on a prescribed function defined on the information space.

Let
$
f:\mathbb{F}_{2}^{h\times k}\longrightarrow {\rm Im}(f)
$
be an arbitrary function. This function naturally induces an equivalence relation on the information matrices by declaring two matrices to be equivalent whenever they have the same function value. More precisely, for every
$A,B\in\mathbb{F}_{2}^{h\times k}$,
we define
$
A\sim_f B
\quad\Longleftrightarrow\quad
f(A)=f(B).
$ This equivalence relation induces an equivalence class for every $A\in\mathbb{F}_{2}^{h\times k}$, which we denote by $$[A]_f=\{B\in\mathbb{F}_{2}^{h\times k}\,|\,f(B)=f(A)\}.$$This means that for every $A,B\in\mathbb{F}_{2}^{h\times k}$, $$[A]_f\neq[B]_f
\quad\Longleftrightarrow\quad
f(A)\neq f(B).$$

    Accordingly, unlike classical error-correcting codes, function-correcting RT codes do not aim to distinguish among all information matrices. Instead, they need only distinguish the equivalence classes induced by the function $f$. Consequently, after transmission through parallel channels, the decoder is only required to recover the correct function value, rather than the complete information matrix. This relaxation of the decoding requirement is precisely what enables the reduction of coding redundancy.

We are now ready to formally introduce the notion of function-correcting RT codes. Next, we introduce the definition.

\begin{Definition}\label{Def FCRTCs}
Suppose that the RT-distance between the transmitted codeword and the received word is at most $t$, where $h$ and $k$ are positive integers. Let
$$
f:\mathbb{F}_{2}^{h\times k}\longrightarrow{\rm Im}(f)
$$
be an arbitrary function. Consider a systematic encoding function
$$
Enc:\mathbb{F}_{2}^{h\times k}
\longrightarrow
\mathbb{F}_{2}^{h\times(k+r)},
$$
defined by
$
Enc(A)=(A,p(A)),
\,
A\in\mathbb{F}_{2}^{h\times k},
$
where $p(A)$ denotes the redundancy matrix associated with $A$.

If, for every pair
$A_1,A_2\in\mathbb{F}_{2}^{h\times k}$ satisfying
$
[A_1]_f\neq[A_2]_f,
$
the corresponding codewords satisfy
$$
d_{RT}\bigl(Enc(A_1),Enc(A_2)\bigr)\ge 2t+1,
$$
then
$
\{\,Enc(A)\mid A\in\mathbb{F}_{2}^{h\times k}\,\}
$
is called a \emph{function-correcting RT code} for the function $f$, abbreviated as an \emph{FCRTC}.
\end{Definition}

\begin{Remark}
When $h=1$, the RT metric reduces to the classical Hamming metric. Consequently, an FCRTC coincides with a function-correcting code in the sense of Lenz \emph{et al.}~\cite{lbwy}. Therefore, FCRTCs constitute a genuine extension of the original function-correcting coding framework from the Hamming metric to the more general Rosenbloom--Tsfasman metric.
\end{Remark}

The decoding principle underlying Definition~\ref{Def FCRTCs} is based on the maximum-likelihood decoder with respect to the RT metric. Provided that the RT-distance between the transmitted codeword and the received word does not exceed $t$, the receiver can uniquely determine the correct function value, assuming that both the encoding rule $Enc(\cdot)$ and the function $f(\cdot)$ are known.

It is important to emphasize that the decoder is not required to recover the original information matrix itself. Instead, its sole objective is to identify the correct equivalence class induced by the function $f$, or equivalently, to recover the correct value of $f(A)$. This distinction is precisely what enables function-correcting RT codes to achieve a smaller redundancy than classical error-correcting codes.

The central objective of this paper is therefore to determine the minimum redundancy required to construct FCRTCs with a prescribed error-correction capability.

     The redundancy introduced by an FCRTC takes the form of an additional matrix appended to the information matrix through systematic encoding. Since the number of matrix rows is fixed throughout the paper, the efficiency of an FCRTC is naturally measured by the minimum number of additional columns required to guarantee the prescribed error-correction capability. This leads to the following fundamental parameter, which will be the main object of investigation in the remainder of this paper.

\begin{Definition}\label{Def r_{RT}^{f}(h,k,t)}
Let $h$ and $k$ be positive integers. The \emph{optimal redundancy} of FCRTCs for a function
$f:\mathbb{F}_{2}^{h\times k}\rightarrow{\rm Im}(f)$,
denoted by
$
r_{RT}^{f}(h,k,t),
$
is defined as the smallest integer $r$ for which there exists a function-correcting RT code with a systematic encoding function
$
Enc:\mathbb{F}_{2}^{h\times k}
\longrightarrow
\mathbb{F}_{2}^{h\times(k+r)}
$
for the function $f$ over $\mathbb{F}_{2}^{h\times k}$.
\end{Definition}

\begin{Remark}
The definition is consistent with the intuition that the complexity of the target function influences the amount of redundancy required. In the extreme case where $f$ is a constant function over $\mathbb{F}_{2}^{h\times k}$, every information matrix belongs to the same equivalence class. Consequently, no redundancy is needed to distinguish between different function values, and therefore $r_{RT}^{f}(h,k,t)= 0.$
\end{Remark}

Determining the exact value of $r_{RT}^{f}(h,k,t)$ is the central problem addressed in this paper. As will become apparent in the subsequent sections, this optimization problem can be reformulated in terms of a new family of matrix codes endowed with irregular RT-distance constraints. This alternative viewpoint provides a unified framework for deriving both general bounds and explicit constructions.

For convenience, we introduce the following notation, which will be used throughout the remainder of the paper. For any integer $a$, let
\[
[a]^+ \triangleq \max\{a,0\}.
\]
For every positive integer $M$, let
\[
[M]\triangleq\{1,2,\ldots,M\}.
\]
For a matrix $D$, we denote by $[D]_{ij}$ its $(i,j)$-th entry. Moreover,
$\mathbb{N}_0^{M\times M}$ denotes the set of all $M\times M$ matrices with nonnegative integer entries.

    The systematic structure of FCRTCs naturally separates each codeword into an information part and a redundancy part. Consequently, the problem of determining the optimal redundancy can be reformulated as constructing a family of redundancy matrices that satisfy suitable pairwise RT-distance constraints. This observation motivates the introduction of the following class of matrix codes, which plays a central role throughout the paper.

    \begin{Definition}\label{Def D-RT code}
Let $h,\,M$ be positive integers and let
\[
D\in\mathbb{N}_0^{M\times M}
\]
be a matrix satisfying
\[
[D]_{ii}=0,\qquad i\in[M].
\]

Let
\[
\mathcal P=\{P_1,P_2,\ldots,P_M\}
\subseteq
\mathbb F_2^{h\times r}
\]
be a code of cardinality $M$.

If there exists an ordering of the codewords of $\mathcal P$ such that
\[
d_{RT}(P_i,P_j)\ge [D]_{ij},
\qquad
\forall\, i,j\in[M],
\]
then $\mathcal P$ is called an \emph{irregular RT-distance code}, or simply a \emph{$D$-RT code}.

Furthermore, for a given matrix $D$ and a fixed integer $h$, the smallest integer $r$ for which there exists a $D$-RT code over
$\mathbb F_2^{h\times r}$
is called the \emph{smallest column length} of $D$-RT codes and is denoted by
\[
N_{RT}^{h}(D).
\]
\end{Definition}

\begin{Remark}
Unlike classical codes, whose distance requirements are uniform for every pair of distinct codewords, a $D$-RT code allows different pairs of codewords to satisfy different prescribed RT-distance constraints. The matrix $D$ therefore completely specifies the geometric distance requirements imposed on the code and provides considerable flexibility in the design of matrix codes. As will be shown in the next section, this generalized framework is particularly well suited for characterizing the optimal redundancy of FCRTCs.
\end{Remark}

The next step is to determine how the distance requirements imposed by a function can be encoded in a matrix of prescribed pairwise RT distances. This motivates the introduction of RT-distance matrices and RT-function-distance matrices.  
    \begin{Definition}\label{Def RT-distance matrix}
       Let $f:\mathbb{F}_{2}^{h \times k} \rightarrow {\rm Im}(f)$ be a function and $ A_{1},A_{2},\ldots,A_{M} \in \mathbb{F}_{2}^{h \times k} $ be $ M $ matrices over $ \mathbb{F}_{2}^{h \times k} $, where $ h,k,M $ are positive integers. We define the \emph{RT-distance matrix} $ D_{RT}^{f}(h,k,t,A_{1},A_{2},\ldots,A_{M}) $ for the function $f$ as a $ M \times M $ matrix with entries
       \begin{equation}\label{Formula 2}
           [D_{RT}^{f}(h,k,t,A_{1},A_{2},\ldots,A_{M})]_{ij}=\left\{
  	         \begin{aligned}
  	          &[2t+1-d_{RT}(A_{i},A_{j})]^{+},& \text{if } f(A_{i}) \neq f(A_{j}), \\
                &0,& \text{otherwise}.
               \end{aligned}
               \right.
        \end{equation}
    \end{Definition}

    For a given function, FCRTCs adopt the system encoding method; thus, we can use RT-distance matrices to connect the optimal redundancy of FCRTCs with the smallest column length of $D$-RT codes.

    \begin{Definition}\label{Def RT-function-distance}
      Let $ f:\mathbb{F}_{2}^{h \times k} \rightarrow {\rm Im}(f) $ be a function. For any $ z \in {\rm Im}(f), $ let $ f^{-1}(z) $ be the set of preimages of $ z, $ i.e. , $ f^{-1}(z)={ \{\, A \in \mathbb{F}_{2}^{h \times k} \, | \, f(A)=z \, \} }. $ For $ \forall z_{1}, \, z_{2} \in {\rm Im}(f), $ define the \emph{RT-function-distance} between function values $ d_{RT}^{f}(z_{1},z_{2}) $ as follows.
      $$ d_{RT}^{f}(z_{1},z_{2})={ \min{ \{\, d_{RT}(A,B) \, | \, A \in f^{-1}(z_{1}) , \, B \in f^{-1}(z_{2}) \, \} }}, $$
      specifically, when $ z_{1}=z_{2} \in {\rm Im}(f) $, we have $ d_{RT}^{f}(z_{1},z_{2})=0 $.
    \end{Definition}

    Using Definition \ref{Def RT-function-distance}, we give the definition of the RT-function-distance matrix for a function $f$ as follows.
    \begin{Definition}\label{Def RT-function-distance matrix}
      Let $ f:\mathbb{F}_{2}^{h \times k} \rightarrow {\rm Im}(f) $ be a function, where $|{\rm Im}(f)|=E$ and ${\rm Im}(f)=\{f_{1},\ldots,f_{E}\}$. The \emph{RT-function-distance matrix} for the function $f$ is denoted by an $E \times E$ matrix $ D_{RT}^{f}(h,k,t,f_{1},\ldots,f_{E}), $ whose entries are defined as follows.
       \begin{equation}\label{Formula 3}
           [D_{RT}^{f}(h,k,t,f_{1},\ldots,f_{E})]_{ij}=\left\{
  	         \begin{aligned}
  	          &[2t+1-d_{RT}^{f}(f_{i},f_{j})]^{+},& \text{if } i \neq j , \\
                &0,& \text{otherwise}.
               \end{aligned}
               \right.
        \end{equation}
    \end{Definition}

\section{Bounds of FCRTCs for General Functions}

In this section, we investigate the fundamental problem of determining the optimal redundancy of function-correcting RT codes associated with an arbitrary function
\[
f:\mathbb{F}_2^{h\times k}\longrightarrow{\rm Im}(f),
\]
where
\[
E=|{\rm Im}(f)|.
\]

Our approach proceeds in two successive steps. We first establish an exact correspondence between the optimal redundancy
$r_{RT}^{f}(h,k,t)$
of FCRTCs and the smallest column length of appropriately constructed irregular RT-distance codes. This correspondence transforms the redundancy optimization problem into a purely geometric problem on matrix codes with prescribed pairwise RT-distance constraints.

Building upon this characterization, we then derive general upper and lower bounds on the smallest column length
$N_{RT}^{h}(D)$
of $D$-RT codes. These results immediately translate into corresponding bounds on the optimal redundancy
$r_{RT}^{f}(h,k,t)$
for arbitrary functions. Consequently, this section provides the general theoretical framework upon which all subsequent developments of the paper are based.

\subsection{Connections between FCRTCs and irregular RT-distance codes}

The main objective of this subsection is to establish a precise connection between function-correcting RT codes and irregular RT-distance codes. This connection constitutes the cornerstone of our approach, since it enables the optimal redundancy problem for FCRTCs to be reformulated as a coding-theoretic problem involving matrix codes with prescribed pairwise RT-distance constraints.

The key observation is that FCRTCs employ systematic encoding. Consequently, every codeword naturally decomposes into an information matrix and a redundancy matrix, while the information part is completely determined by the transmitted message. Therefore, the distance constraints required for successful function correction can be entirely transferred to the redundancy matrices.

To formalize this idea, we make use of the RT-distance matrices and RT-function-distance matrices introduced in the previous section. These matrices encode the RT-distance requirements imposed by the function and enable us to construct suitable $D$-RT codes. As a consequence, we establish exact relationships between the optimal redundancy
$r_{RT}^{f}(h,k,t)$
and the smallest column length
$N_{RT}^{h}(D)$
of the corresponding $D$-RT codes.

The following theorem establishes the fundamental link between function-correcting RT codes and irregular RT-distance codes. Its proof relies on the systematic structure of FCRTCs. More precisely, we show that the RT-distance constraints required for successful function correction can be transferred exactly to the redundancy matrices, thereby establishing a one-to-one correspondence between FCRTCs and suitable $D$-RT codes. It shows that the problem of determining the optimal redundancy of an FCRTC is exactly equivalent to determining the smallest column length of an appropriately constructed $D$-RT code. Consequently, the redundancy optimization problem is transformed into a geometric problem involving prescribed pairwise RT-distance constraints.

	\begin{Theorem}\label{Th1}
Let
\[
f:\mathbb{F}_2^{h\times k}\rightarrow{\rm Im}(f)
\]
be an arbitrary function, and let
\[
\{A_1,A_2,\ldots,A_{2^{hk}}\}
=
\mathbb{F}_2^{h\times k}.
\]
Then
\begin{equation}\label{Formula 5}
r_{RT}^{f}(h,k,t)
=
N_{RT}^{h}
\!\left(
D_{RT}^{f}(h,k,t,A_1,A_2,\ldots,A_{2^{hk}})
\right).
\end{equation}
\end{Theorem}

\begin{proof}
For simplicity, let
\[
r_1=r_{RT}^{f}(h,k,t),
\qquad
r_2=N_{RT}^{h}\!\left(D_{RT}^{f}(h,k,t,A_{1},A_{2},\ldots,A_{2^{hk}})\right).
\]

We prove the theorem by establishing the two inequalities
$r_1\le r_2$ and $r_1\ge r_2$.

We first prove that $r_1\leq r_2$. Let
\[
\{P_1,P_2,\ldots,P_{2^{hk}}\}
\]
be a
$D_{RT}^{f}(h,k,t,A_{1},A_{2},\ldots,A_{2^{hk}})$-RT code over
$\mathbb{F}_2^{h\times r_2}$.
By Definition~\ref{Def D-RT code},
\[
d_{RT}(P_i,P_j)
\geq
\bigl[D_{RT}^{f}(h,k,t,A_{1},A_{2},\ldots,A_{2^{hk}})\bigr]_{ij},
\qquad
\forall\, i,j\in[2^{hk}].
\]

Define the systematic encoding function
\[
Enc:\mathbb F_2^{h\times k}
\longrightarrow
\mathbb F_2^{h\times(k+r_2)},
\qquad
A_i\longmapsto(A_i,P_i).
\]

Consider two matrices
$A_i,A_j\in\mathbb F_2^{h\times k}$
such that
$f(A_i)\neq f(A_j)$.
Since the encoding is systematic,
\[
d_{RT}(Enc(A_i),Enc(A_j))
=
d_{RT}(A_i,A_j)
+
d_{RT}(P_i,P_j).
\]

By Definition~\ref{Def RT-distance matrix},
$
[D_{RT}^{f}(h,k,t,A_{1},A_{2},\ldots,A_{2^{hk}})]_{ij}
\ge
2t+1-d_{RT}(A_i,A_j),
$
and therefore
\[
\begin{aligned}
d_{RT}(Enc(A_i),Enc(A_j))
&\ge
d_{RT}(A_i,A_j)
+
2t+1-d_{RT}(A_i,A_j)  \\
&=2t+1.
\end{aligned}
\]

Hence,
$
\{\,Enc(A_i)\mid
A_i\in\mathbb F_2^{h\times k}\,\}
$
is an FCRTC. Consequently,
$
r_1\le r_2.
$

We now prove the converse inequality. Assume, for contradiction, that
$r_1<r_2$.

Let
\[
Enc:
\mathbb F_2^{h\times k}
\longrightarrow
\mathbb F_2^{h\times(k+r_1)},
\qquad
A_i\longmapsto(A_i,P_i),
\]
be an FCRTC.

Since
$r_1<r_2$,
the collection
\[
\{P_1,P_2,\ldots,P_{2^{hk}}\}
\]
cannot form a
$D_{RT}^{f}(h,k,t,A_{1},A_{2},\ldots,A_{2^{hk}})$-RT code.
Hence there exist
$i_0\neq j_0$
such that
\[
f(A_{i_0})\neq f(A_{j_0})
\]
and
\[
d_{RT}(P_{i_0},P_{j_0})
<
2t+1-
d_{RT}(A_{i_0},A_{j_0}).
\]

Therefore,
\[
\begin{aligned}
d_{RT}(Enc(A_{i_0}),Enc(A_{j_0}))
&=
d_{RT}(A_{i_0},A_{j_0})
+
d_{RT}(P_{i_0},P_{j_0})\\
&<
2t+1,
\end{aligned}
\]
which contradicts Definition~\ref{Def FCRTCs}.

Thus,
\[
r_1\ge r_2.
\]

Combining the two inequalities yields
\[
r_{RT}^{f}(h,k,t)
=
N_{RT}^{h}
\!\left(
D_{RT}^{f}(h,k,t,A_{1},A_{2},\ldots,A_{2^{hk}})
\right),
\]
thereby completing the proof.
\end{proof}

Although Theorem~\ref{Th1} provides an exact characterization of the optimal redundancy, its direct application may become computationally prohibitive when $h$ and $k$ are large, since the associated RT-distance matrix involves all $2^{hk}$ information matrices. To overcome this difficulty, we derive the following corollary, which provides a general lower bound on $r_{RT}^{f}(h,k,t)$ by considering only an arbitrary subset of information matrices. This considerably reduces the computational complexity while still yielding meaningful bounds on the optimal redundancy.

\begin{Corollary}\label{Cor1}
Let
\[
A_{1},A_{2},\ldots,A_{M}\in\mathbb{F}_2^{h\times k}
\]
be distinct matrices. Then, for every function
$f:\mathbb{F}_2^{h\times k}\rightarrow{\rm Im}(f)$,
the optimal redundancy of FCRTCs satisfies
\[
r_{RT}^{f}(h,k,t)
\ge
N_{RT}^{h}\!\left(
D_{RT}^{f}(h,k,t,A_{1},A_{2},\ldots,A_{M})
\right).
\]

Furthermore, if
\[
|{\rm Im}(f)|\ge 2^{k(h-1)}+1,
\]
then
\[
r_{RT}^{f}(h,k,t)
\ge
\left\lceil\frac{2t}{h}\right\rceil.
\]
\end{Corollary}

\begin{proof}
For simplicity, let
\[
r_1=r_{RT}^{f}(h,k,t),
\qquad
r_2=
N_{RT}^{h}\!\left(
D_{RT}^{f}(h,k,t,A_{1},A_{2},\ldots,A_{M})
\right).
\]

By Theorem~\ref{Th1},
\[
r_1=
N_{RT}^{h}\!\left(
D_{RT}^{f}(h,k,t,A_{1},A_{2},\ldots,A_{2^{hk}})
\right),
\]
where
\[
\mathbb{F}_2^{h\times k}
=
\{A_1,A_2,\ldots,A_{2^{hk}}\}.
\]

Let
\[
\{P_1,P_2,\ldots,P_{2^{hk}}\}
\]
be a
$D_{RT}^{f}(h,k,t,A_{1},A_{2},\ldots,A_{2^{hk}})$-RT code over
$\mathbb{F}_2^{h\times r_1}$.
By Definition~\ref{Def D-RT code},
\[
d_{RT}(P_i,P_j)
\ge
[D_{RT}^{f}(h,k,t,A_{1},A_{2},\ldots,A_{2^{hk}})]_{ij},
\qquad
\forall\,i,j\in[2^{hk}].
\]

Since
\[
\{A_1,A_2,\ldots,A_M\}
\subseteq
\mathbb{F}_2^{h\times k},
\]
the matrix
\[
D_{RT}^{f}(h,k,t,A_{1},A_{2},\ldots,A_M)
\]
is a principal submatrix of
\[
D_{RT}^{f}(h,k,t,A_{1},A_{2},\ldots,A_{2^{hk}}).
\]
Consequently,
\[
\{P_1,P_2,\ldots,P_M\}
\]
forms a
$D_{RT}^{f}(h,k,t,A_{1},A_{2},\ldots,A_M)$-RT code over
$\mathbb{F}_2^{h\times r_1}$.
Therefore,
\[
r_{RT}^{f}(h,k,t)
\ge
N_{RT}^{h}
\!\left(
D_{RT}^{f}(h,k,t,A_{1},A_{2},\ldots,A_M)
\right).
\]

We now prove the second assertion.

Assume that
\[
|{\rm Im}(f)|
\ge
2^{(h-1)k}+1.
\]

We first claim that there exist two matrices
\[
A_1,A_2\in\mathbb{F}_2^{h\times k}
\]
such that
\[
d_{RT}(A_1,A_2)=1
\]
and
\[
f(A_1)\neq f(A_2).
\]

Suppose, to the contrary, that
\[
f(A_1)=f(A_2)
\]
for every pair of matrices satisfying
\[
d_{RT}(A_1,A_2)=1.
\]

Define
\[
S=
\left\{
A\in\mathbb{F}_2^{h\times k}
\;\middle|\;
\text{the last row of }A\text{ is }(0,0,\ldots,0)
\right\}.
\]

For every
$B\in S$, let
\[
S_B=
\left\{
A\in\mathbb{F}_2^{h\times k}
\;\middle|\;
\text{the first }(h-1)\text{ rows of }A
\text{ coincide with those of }B
\right\},
\]
and
\[
B_{RT}(B,1)
=
\left\{
A\in\mathbb{F}_2^{h\times k}
\;\middle|\;
d_{RT}(A,B)=1
\right\}.
\]

Clearly,
\[
\mathbb{F}_2^{h\times k}
=
\bigcup_{B\in S}S_B.
\]

Fix an arbitrary matrix
$B_0\in S$.
For every
\[
B_1\in B_{RT}(B_0,1),
\]
our assumption implies that
\[
f(B_1)=f(B_0).
\]

Similarly, for every
\[
B_2\in B_{RT}(B_1,1),
\]
we obtain
\[
f(B_2)=f(B_1)=f(B_0).
\]

Iterating this argument, every matrix belonging to
$S_{B_0}$
has the same function value as
$B_0$.
Consequently, each set
$S_B$
corresponds to a single function value.

Since
\[
|S|
=
2^{(h-1)k},
\]
it follows that
\[
|{\rm Im}(f)|
\le
2^{(h-1)k},
\]
contradicting the hypothesis.
Therefore, there exist
$A_1,A_2\in\mathbb{F}_2^{h\times k}$
such that
\[
d_{RT}(A_1,A_2)=1
\]
and
\[
f(A_1)\neq f(A_2).
\]

By the first part of the proof,
\[
r_{RT}^{f}(h,k,t)
\ge
N_{RT}^{h}
\!\left(
D_{RT}^{f}(h,k,t,A_1,A_2)
\right),
\]
where
\[
D_{RT}^{f}(h,k,t,A_1,A_2)
=
\begin{pmatrix}
0&2t\\
2t&0
\end{pmatrix}.
\]

Consider the code
$
\mathcal P=\{P_1,P_2\}
$
over
$
\mathbb{F}_2^{h\times
\left\lceil\frac{2t}{h}\right\rceil},
$
where
\[
P_1=
\begin{pmatrix}
0&0&\cdots&0\\
0&0&\cdots&0\\
\vdots&\vdots&\ddots&\vdots\\
0&0&\cdots&0
\end{pmatrix},
\qquad
P_2=
\begin{pmatrix}
1&1&\cdots&1\\
1&1&\cdots&1\\
\vdots&\vdots&\ddots&\vdots\\
1&1&\cdots&1
\end{pmatrix}.
\]

Since
\[
d_{RT}(P_1,P_2)
=
h
\left\lceil
\frac{2t}{h}
\right\rceil
\ge
2t,
\]
Definition~\ref{Def D-RT code} implies that
$\mathcal P$
is a
$D_{RT}^{f}(h,k,t,A_1,A_2)$-RT code.
Hence,
\[
N_{RT}^{h}
\!\left(
D_{RT}^{f}(h,k,t,A_1,A_2)
\right)
\le
\left\lceil
\frac{2t}{h}
\right\rceil.
\]

Conversely, let
\[
\mathcal P'
=
\{P'_1,P'_2\}
\]
be any
$D_{RT}^{f}(h,k,t,A_1,A_2)$-RT code over
$\mathbb{F}_2^{h\times m}$.
By Definition~\ref{Def D-RT code},
$
d_{RT}(P'_1,P'_2)
\ge
2t.
$
On the other hand, by Definition~\ref{Def1 RT-weight},
$
d_{RT}(P'_1,P'_2)
\le
hm.
$
Therefore,
$
m
\ge
\frac{2t}{h}.
$
Since $m$ is an integer,
$
m
\ge
\left\lceil
\frac{2t}{h}
\right\rceil.
$
Consequently,
$
N_{RT}^{h}
\!\left(
D_{RT}^{f}(h,k,t,A_1,A_2)
\right)
=
\left\lceil
\frac{2t}{h}
\right\rceil.
$

Finally,
$
r_{RT}^{f}(h,k,t)
\ge
\left\lceil
\frac{2t}{h}
\right\rceil,
$
which completes the proof.
\end{proof}

   \begin{Remark}
Corollary~\ref{Cor1} holds for arbitrary positive integers $h$ and $k$. In the particular case where $h=1$, the RT metric coincides with the classical Hamming metric. Consequently, Corollary~\ref{Cor1} reduces exactly to Corollary~1 of~\cite{lbwy}. This further illustrates that the present framework constitutes a natural extension of the classical theory of function-correcting codes to the Rosenbloom--Tsfasman metric.
\end{Remark}

The lower bound established in Corollary~\ref{Cor1} is derived from the RT-distance matrix associated with the information matrices. We now turn to the complementary problem of deriving a general upper bound on the optimal redundancy. Instead of considering all information matrices individually, we exploit the RT-function-distance matrix, which captures only the distance requirements between distinct function values. This significantly reduces the size of the underlying matrix and yields a general upper bound for $r_{RT}^{f}(h,k,t)$.

\begin{Theorem}\label{Th2}
Let 
\[
f:\mathbb{F}_2^{h\times k}\longrightarrow{\rm Im}(f)
\]
be a function. Then
\begin{equation}\label{Formula 6}
r_{RT}^{f}(h,k,t)
\le
N_{RT}^{h}
\!\left(
D_{RT}^{f}(h,k,t,f_{1},\ldots,f_{E})
\right),
\end{equation}
where
$
E=|{\rm Im}(f)|
$
and
$
{\rm Im}(f)=\{f_{1},f_{2},\ldots,f_{E}\}.
$
\end{Theorem}

\begin{proof}
Let
\[
{\rm Im}(f)=\{f_1,f_2,\ldots,f_E\},
\qquad
E=|{\rm Im}(f)|.
\]
For simplicity, let
\[
r_1=r_{RT}^{f}(h,k,t),
\qquad
r_2=
N_{RT}^{h}\!\left(
D_{RT}^{f}(h,k,t,f_1,\ldots,f_E)
\right).
\]

Let
\[
\{P_1,P_2,\ldots,P_E\}
\]
be a
$D_{RT}^{f}(h,k,t,f_1,\ldots,f_E)$-RT code over
$\mathbb{F}_2^{h\times r_2}$.
By Definition~\ref{Def D-RT code},
\[
d_{RT}(P_i,P_j)
\ge
[D_{RT}^{f}(h,k,t,f_1,\ldots,f_E)]_{ij},
\qquad
\forall\,i,j\in[E].
\]

Moreover, by Definition~\ref{Def RT-function-distance},
\[
d_{RT}(P_i,P_j)
\ge
2t+1-d_{RT}(f_i,f_j),
\qquad
\forall\, i\neq j.
\]

We now define a systematic encoding function
\[
Enc:
\mathbb{F}_2^{h\times k}
\longrightarrow
\mathbb{F}_2^{h\times(k+r_2)},
\]
by
\[
Enc(A)=(A,P_i),
\]
where
\[
f(A)=f_i.
\]

Consider arbitrary matrices
\[
A,B\in\mathbb{F}_2^{h\times k},
\]
and suppose that
\[
f(A)=f_i,
\qquad
f(B)=f_j.
\]
If
$f_i\neq f_j$,
then
\[
\begin{aligned}
d_{RT}(Enc(A),Enc(B))
&=
d_{RT}(A,B)
+
d_{RT}(P_i,P_j)\\
&\ge
d_{RT}(f_i,f_j)
+
2t+1-d_{RT}(f_i,f_j)\\
&=
2t+1,
\end{aligned}
\]
where the first inequality follows from Definition~\ref{Def RT-function-distance}.

Hence, Definition~\ref{Def FCRTCs} implies that
$
\{\,Enc(A)\mid A\in\mathbb{F}_2^{h\times k}\,\}
$
is an FCRTC over
$\mathbb{F}_2^{h\times(k+r_2)}$.
Therefore,
$
r_1\le r_2,
$
that is,
$
r_{RT}^{f}(h,k,t)
\le
N_{RT}^{h}
\!\left(
D_{RT}^{f}(h,k,t,f_1,\ldots,f_E)
\right).
$
This completes the proof.
\end{proof}

The previous results provide complementary estimates for the optimal redundancy of FCRTCs: Corollary~\ref{Cor1} establishes a general lower bound, whereas Theorem~\ref{Th2} provides a corresponding upper bound. A natural question is therefore to determine when these two bounds coincide.

Recall that the entries of an RT-function-distance matrix are defined in terms of the RT distances between suitable information matrices. Consequently, whenever the RT-function-distance matrix can be realized by a family of representative information matrices, the lower and upper bounds become identical. The following corollary identifies such a situation and provides an exact characterization of the optimal redundancy.
 
 \begin{Corollary}\label{Cor2}
Let
\[
f:\mathbb{F}_2^{h\times k}\rightarrow{\rm Im}(f)
\]
be a function and $E=|{\rm Im}(f)|$. Suppose that there exists a family of representative information matrices
\[
A_1,A_2,\ldots,A_E
\]
such that
\[
\{f(A_1),f(A_2),\ldots,f(A_E)\}
=
{\rm Im}(f)
\]
and
\[
D_{RT}^{f}(h,k,t,f_1,\ldots,f_E)
=
D_{RT}^{f}(h,k,t,A_1,\ldots,A_E).
\]
Then
\[
r_{RT}^{f}(h,k,t)
=
N_{RT}^{h}
\!\left(
D_{RT}^{f}(h,k,t,f_1,\ldots,f_E)
\right).
\]
\end{Corollary}

  \begin{proof}
By Corollary~\ref{Cor1},
\[
r_{RT}^{f}(h,k,t)
\ge
N_{RT}^{h}
\!\left(
D_{RT}^{f}(h,k,t,A_1,\ldots,A_E)
\right).
\]
Since
\[
D_{RT}^{f}(h,k,t,A_1,\ldots,A_E)
=
D_{RT}^{f}(h,k,t,f_1,\ldots,f_E),
\]
we obtain
\[
r_{RT}^{f}(h,k,t)
\ge
N_{RT}^{h}
\!\left(
D_{RT}^{f}(h,k,t,f_1,\ldots,f_E)
\right).
\]
On the other hand, Theorem~\ref{Th2} yields the reverse inequality. Therefore,
\[
r_{RT}^{f}(h,k,t)
=
N_{RT}^{h}
\!\left(
D_{RT}^{f}(h,k,t,f_1,\ldots,f_E)
\right),
\]
which completes the proof.
\end{proof}

\subsection{Bounds on $N_{RT}^{h}(D)$ for Certain Classes of Matrices}

The results established in the previous subsection reduce the problem of determining the optimal redundancy of FCRTCs to the problem of computing the smallest column length
$
N_{RT}^{h}(D)
$
of suitably defined $D$-RT codes. Consequently, deriving bounds on
$N_{RT}^{h}(D)$
immediately yields corresponding bounds on the optimal redundancy
$r_{RT}^{f}(h,k,t)$
for arbitrary functions.

Motivated by this observation, we now investigate upper and lower bounds on
$N_{RT}^{h}(D)$
for several classes of matrices $D$. These results constitute the main technical ingredients for the general redundancy bounds established later in this section.

We begin by recalling a useful result of Rosenbloom and Tsfasman~\cite{rt}, which will play a key role in deriving our subsequent lower bounds.

\begin{Lemma}\label{Lem2}{\rm(\cite{rt})}
Let
\[
X=(x_1,x_2,\ldots,x_h)^{\intercal},
\qquad
Y=(y_1,y_2,\ldots,y_h)^{\intercal}
\in\mathbb{F}_2^{h},
\]
and let $K$ be a positive constant.

Define
\[
d(X,Y)
=
h-
\max\{
\,i
\mid
x_1=y_1,\,
x_2=y_2,\,
\ldots,\,
x_i=y_i
\},
\]
and
\[
F(t)
=
\sum_{X,Y\in\mathbb{F}_2^{h}}
d(X,Y)t_Xt_Y,
\]
where
$t_X$
is a nonnegative real number for every
$X\in\mathbb{F}_2^{h}$
satisfying
\[
\sum_{X\in\mathbb{F}_2^{h}}t_X=K.
\]

Then the maximum value of $F(t)$ is
\[
K^2\left(h+2^{-h}-1\right).
\]
\end{Lemma}

Observe that the distance $d(X,Y)$ introduced in Lemma~\ref{Lem2} coincides with the RT distance defined in Definition~\ref{def RT-distance}. This observation enables us to exploit the extremal property established in Lemma~\ref{Lem2} to derive a general lower bound on the smallest column length of $D$-RT codes. The resulting estimate holds for every nonnegative symmetric matrix.

\begin{Lemma}\label{Lem3}
Let $h$ and $M$ be positive integers. Let
\[
D\in\mathbb{N}_0^{M\times M}
\]
be a nonnegative symmetric matrix satisfying
\[
[D]_{ii}=0,
\qquad
\forall\,i\in[M].
\]
Then
\[
N_{RT}^{h}(D)
\ge
\frac{2^{h+1}}
{M^2\bigl((h-1)2^h+1\bigr)}
\sum_{\substack{i,j\\ i<j}}
[D]_{ij}.
\]
\end{Lemma}

\begin{proof}
Let
\[
r=N_{RT}^{h}(D),
\]
and let
\[
C=\{P_1,P_2,\ldots,P_M\}
\subseteq
\mathbb{F}_2^{h\times r}
\]
be a $D$-RT code. By Definition~\ref{Def D-RT code},
\[
d_{RT}(P_i,P_j)\ge [D]_{ij},
\qquad
\forall\, i,j\in[M].
\]

For each column position
$j\in[r]$
and every vector
$X\in\mathbb{F}_2^{h\times1}$,
let
$t_j(X)$
denote the number of codewords in $C$ whose $j$th column is equal to $X$.
Since the RT weight is additive over the columns, we have
\[
\sum_{i,j}d_{RT}(P_i,P_j)
=
\sum_{j=1}^{r}
\sum_{X,Y\in\mathbb{F}_2^{h}}
d_{RT}(X,Y)t_j(X)t_j(Y).
\]

Moreover, for every
$j\in[r]$,
\[
\sum_{X\in\mathbb{F}_2^{h}}t_j(X)=M.
\]
Therefore, Lemma~\ref{Lem2} implies that
\[
\sum_{X,Y\in\mathbb{F}_2^{h}}
d_{RT}(X,Y)t_j(X)t_j(Y)
\le
M^2\left(h+2^{-h}-1\right).
\]

Summing over all column positions yields
\[
\sum_{i,j}d_{RT}(P_i,P_j)
\le
rM^2\left(h+2^{-h}-1\right)
=
rM^2\,2^{-h}\bigl((h-1)2^h+1\bigr).
\]

On the other hand, since
\[
d_{RT}(P_i,P_j)\ge [D]_{ij},
\]
we obtain
\[
\sum_{i,j}[D]_{ij}
\le
\sum_{i,j}d_{RT}(P_i,P_j)
\le
rM^2\,2^{-h}\bigl((h-1)2^h+1\bigr).
\]
Hence,
\[
r
\ge
\frac{2^h}
{M^2\bigl((h-1)2^h+1\bigr)}
\sum_{i,j}[D]_{ij}.
\]

Finally, since $D$ is symmetric and satisfies
$[D]_{ii}=0$,
\[
\sum_{i,j}[D]_{ij}
=
2\sum_{i<j}[D]_{ij}.
\]
Consequently,
\[
N_{RT}^{h}(D)
=
r
\ge
\frac{2^{h+1}}
{M^2\bigl((h-1)2^h+1\bigr)}
\sum_{i<j}[D]_{ij},
\]
which completes the proof.
\end{proof}

\begin{Remark}
Lemma~\ref{Lem3} extends the corresponding lower bound established for function-correcting codes under the Hamming metric. Indeed, when $h=1$ and $M$ is odd, the RT metric coincides with the Hamming metric, and Lemma~\ref{Lem3} reduces exactly to Lemma~1 of~\cite{lbwy}.
\end{Remark}
    
    Lemma~\ref{Lem3} provides a general lower bound on the smallest column length
$N_{RT}^{h}(D)$
for arbitrary nonnegative symmetric matrices. We now turn to the complementary problem of deriving an upper bound. To this end, we first introduce the notion of an RT-distance ball and its volume, which will play a fundamental role in the subsequent Gilbert--Varshamov-type argument.

  For
$
A\in\mathbb{F}_2^{h\times k},
$
the \emph{RT-distance ball} centered at $A$ with radius $d$ is defined by
\[
B_{RT}(h,k,d,A)
=
\left\{
B\in\mathbb{F}_2^{h\times k}
\;\middle|\;
d_{RT}(A,B)\le d
\right\}.
\]

The \emph{volume} of the RT-distance ball is denoted by
$V_{RT}(h,k,d)$
and is given by
\[
V_{RT}(h,k,d)
=
|B_{RT}(h,k,d,A)|
=
\left|
\left\{
B\in\mathbb{F}_2^{h\times k}
\;\middle|\;
d_{RT}(A,B)\le d
\right\}
\right|.
\]

The classical Gilbert--Varshamov construction can be naturally extended to irregular RT-distance codes. This leads to the following general upper bound on the smallest column length $N_{RT}^{h}(D)$.

This dependence is reflected by the appearance of an arbitrary permutation in the following bound.

\begin{Lemma}\label{Lem4}
Let $h$ and $M$ be positive integers, and let
\[
D\in\mathbb{N}_{0}^{M\times M}
\]
be a nonnegative symmetric matrix.
Then, for every permutation
\[
\pi:[M]\rightarrow[M],
\]
we have
\[
N_{RT}^{h}(D)
\le
\min_{r\in\mathbb N}
\left\{
r:
2^{hr}>
\max_{j\in[M]}
\sum_{i=1}^{j-1}
V_{RT}
\!\left(
h,r,[D]_{\pi(i)\pi(j)}-1
\right)
\right\}.
\]
\end{Lemma}

\begin{proof}
By Definition~\ref{Def D-RT code}, the ordering of the codewords is arbitrary. Hence, without loss of generality, we may assume that the permutation $\pi$ is the identity map.

Let
\[
D\in\mathbb{N}_{0}^{M\times M}
\]
be a nonnegative symmetric matrix. We construct the desired $D$-RT code sequentially.

Choose an arbitrary matrix
\[
P_1\in\mathbb{F}_2^{h\times r}.
\]

To select the second codeword, all matrices contained in the RT ball
\[
B_{RT}\!\left(h,r,[D]_{12}-1,P_1\right)
\]
must be avoided. Since this ball contains
\[
V_{RT}(h,r,[D]_{12}-1)
\]
matrices, whenever
\[
2^{hr}
>
V_{RT}(h,r,[D]_{12}-1),
\]
there exists a matrix
\[
P_2\in\mathbb{F}_2^{h\times r}
\]
such that
\[
d_{RT}(P_1,P_2)\ge [D]_{12}.
\]

Suppose next that $P_1,\ldots,P_{j-1}$ have already been chosen so that
\[
d_{RT}(P_i,P_\ell)\ge [D]_{i\ell},
\qquad
1\le i<\ell\le j-1.
\]
To choose $P_j$, it suffices to avoid the union of the RT balls
\[
B_{RT}\!\left(h,r,[D]_{ij}-1,P_i\right),
\qquad
1\le i\le j-1.
\]
Therefore, if
\[
2^{hr}
>
\sum_{i=1}^{j-1}
V_{RT}(h,r,[D]_{ij}-1),
\]
then at least one matrix remains available, and we can choose
\[
P_j\in\mathbb{F}_2^{h\times r}
\]
such that
\[
d_{RT}(P_i,P_j)\ge [D]_{ij},
\qquad
1\le i\le j-1.
\]

Consequently, if
\[
2^{hr}
>
\max_{j\in[M]}
\sum_{i=1}^{j-1}
V_{RT}
\!\left(
h,r,[D]_{\pi(i)\pi(j)}-1
\right),
\]
the above greedy procedure successfully constructs
\[
\{P_1,P_2,\ldots,P_M\},
\]
which forms a $D$-RT code over
$\mathbb{F}_2^{h\times r}$.

Hence,
\[
N_{RT}^{h}(D)
\le
\min_{r\in\mathbb N}
\left\{
r:
2^{hr}>
\max_{j\in[M]}
\sum_{i=1}^{j-1}
V_{RT}
\!\left(
h,r,[D]_{\pi(i)\pi(j)}-1
\right)
\right\},
\]
which completes the proof.
\end{proof}

Since both RT-distance matrices and RT-function-distance matrices are, by Definitions~\ref{Def RT-distance matrix} and~\ref{Def RT-function-distance matrix}, nonnegative symmetric matrices, the general bounds established in Lemmas~\ref{Lem3} and~\ref{Lem4} apply directly to these matrices. Combined with the exact relationships derived in Subsection~3.1, these results immediately translate into general lower and upper bounds on the optimal redundancy
$r_{RT}^{f}(h,k,t)$
for an arbitrary function
$f:\mathbb{F}_2^{h\times k}\rightarrow{\rm Im}(f)$.
The resulting estimates constitute the main theoretical tools for analyzing the particular classes of functions investigated in the subsequent sections.

\section{Bounds of FCRTCs for RT-Weight Functions}
In the previous section, we established general lower and upper bounds on the optimal redundancy of function-correcting RT codes for arbitrary functions. We now specialize these results to the RT-weight function, whose particular algebraic structure allows significantly sharper estimates to be derived.

Throughout this section, let
$
f=w_{RT}
$
denote the RT-weight function over
$\mathbb{F}_2^{h\times k}$.
By Definition~\ref{Def1 RT-weight},
$
{\rm Im}(f)=\{0,1,\ldots,hk\},
$
and therefore
$
|{\rm Im}(f)|=hk+1.
$
For convenience, we set
$
E=hk+1.
$

Our objective is to exploit the structural properties of the RT-weight function to refine the general bounds established in Section~3 and obtain substantially tighter estimates of the optimal redundancy of FCRTCs.
To establish the main results of this section, we first recall the shifted modulo operator introduced in~\cite{lbwy}. For positive integers $a$ and $b$, it is defined by
\begin{equation}
a \text{ smod } b
\triangleq
((a-1)\bmod b)+1
\in\{1,2,\ldots,b\}.
\end{equation}

For example,
\[
3\ \text{smod}\ 3=3,
\qquad
5\ \text{smod}\ 3=2.
\]

    \begin{Theorem}\label{Upper-Lower bounds RT-weight}
Let
\[
f:\mathbb{F}_2^{h\times k}\longrightarrow{\rm Im}(f)
\]
be the RT-weight function defined by
\[
f(A)=w_{RT}(A),
\qquad
A\in\mathbb{F}_2^{h\times k}.
\]
Then
\[
N_{RT}^{h}(D_1)
\le
r_{RT}^{f}(h,k,t)
\le
N_{RT}^{h}(D_2),
\]
where
\[
D_1,D_2\in\mathbb N_0^{E\times E},
\]
and, for every
$i,j\in\{0,1,\ldots,hk\}$,
\[
[D_1]_{(i+1)(j+1)}
=
\begin{cases}
\left[
2t+1-|i-j|-(\min\{i,j\}\bmod h)
\right]^+,
& i\neq j,\\[1ex]
0,
& i=j,
\end{cases}
\]
and
\[
[D_2]_{(i+1)(j+1)}
=
\begin{cases}
\left[
2t+1-|i-j|
\right]^+,
& i\neq j,\\[1ex]
0,
& i=j.
\end{cases}
\]
\end{Theorem}

\begin{Remark}
Theorem~\ref{Upper-Lower bounds RT-weight} replaces the abstract matrices in the general bounds of Section~3 with two explicit matrices that depend only on the RT weights. Consequently, estimating the optimal redundancy for the RT-weight function is reduced to determining the smallest column lengths of the explicitly defined $D_1$-RT and $D_2$-RT codes.
\end{Remark}

 \begin{proof}
Let
$
E=hk+1,
$
and write
$
{\rm Im}(f)=\{f_1,f_2,\ldots,f_E\},
$
where
$
f_i=i-1,
\,
i\in[E].
$

We first establish the upper bound.

Recall that
$
d_{RT}(f_i,f_j)
=
\min
\left\{
d_{RT}(A,B)
\,\middle|\,
w_{RT}(A)=f_i,\,
w_{RT}(B)=f_j
\right\}.
$

Let
$
A,B\in\mathbb F_2^{h\times k}
$
satisfy
$
w_{RT}(A)=f_i=i-1,
\,
w_{RT}(B)=f_j=j-1.
$
Since the RT distance satisfies
$
d_{RT}(A,B)
\ge
|w_{RT}(A)-w_{RT}(B)|,
$
we obtain
$
d_{RT}(A,B)
\ge
|i-j|.
$
Taking the minimum over all such pairs $(A,B)$ yields
$
d_{RT}(f_i,f_j)
\ge
|i-j|.
$

Therefore, by Definition~\ref{Def RT-function-distance matrix},
$
[D_{RT}^{f}(h,k,t,f_1,\ldots,f_E)]_{ij}
\le
[D_2]_{ij},
\qquad
\forall\,i,j\in[E].
$

Now let
$
\mathcal P=
\{P_1,P_2,\ldots,P_E\}
$
be a
$D_2$-RT code over
$
\mathbb F_2^{h\times N_{RT}^{h}(D_2)}.
$
By Definition~\ref{Def D-RT code}, we may assume that
$
d_{RT}(P_i,P_j)
\ge
[D_2]_{ij},
\,
\forall\,i,j\in[E].
$
Consequently,
\[
d_{RT}(P_i,P_j)
\ge
[D_{RT}^{f}(h,k,t,f_1,\ldots,f_E)]_{ij},
\]
which shows that
$
\mathcal P
$
is also a
$
D_{RT}^{f}(h,k,t,f_1,\ldots,f_E)
\text{-RT code}.
$
Hence,
\[
N_{RT}^{h}
\left(
D_{RT}^{f}(h,k,t,f_1,\ldots,f_E)
\right)
\le
N_{RT}^{h}(D_2).
\]

Finally, Theorem~\ref{Th2} yields
$
r_{RT}^{f}(h,k,t)
\le
N_{RT}^{h}
\!\left(
D_{RT}^{f}(h,k,t,f_1,\ldots,f_E)
\right)
\le
N_{RT}^{h}(D_2).
$

The RT weight satisfies the reverse triangle inequality
$
|w_{RT}(A)-w_{RT}(B)|
\le
d_{RT}(A,B),
$
which immediately implies that $
d_{RT}(A,B)
\ge
|w_{RT}(A)-w_{RT}(B)|.
$

 We now establish the lower bound by constructing an explicit family of representative information matrices. More precisely, we select one matrix for each possible RT weight and then apply Corollary~\ref{Cor1}. This construction enables us to derive an explicit lower bound through the associated RT-distance matrix.

Let
\[
A_0=0_{h\times k}
\]
be the zero matrix in
$\mathbb{F}_2^{h\times k}$.
For
\[
1\le i\le hk,
\]
define
\[
A_i=
\begin{pmatrix}
\overset{\big\lceil\frac{i}{h}\big\rceil-1}{\overbrace{1\ \cdots\ 1}}
&0&0&\cdots&0\\
0\ \cdots\ 0&0&0&\cdots&0\\
\vdots&\vdots&\vdots&\ddots&\vdots\\
0\ \cdots\ 0&1&0&\cdots&0\\
\vdots&\vdots&\vdots&\ddots&\vdots\\
0\ \cdots\ 0&0&0&\cdots&0
\end{pmatrix}
\in
\mathbb{F}_2^{h\times k},
\]
where the unique nonzero entry in the
$\left(h-(i\ \mathrm{smod}\ h)+1\right)$-th row
of the
$\left\lceil\frac{i}{h}\right\rceil$-th column
is equal to~$1$.

By construction,
\[
w_{RT}(A_i)=i,
\qquad
0\le i\le hk.
\]
Hence,
\[
\{w_{RT}(A_0),w_{RT}(A_1),\ldots,w_{RT}(A_{hk})\}
={\rm Im}(f).
\]

Therefore, Corollary~\ref{Cor1} implies that
\[
r_{RT}^{f}(h,k,t)
\ge
N_{RT}^{h}
\!\left(
D_{RT}^{f}(h,k,t,A_0,A_1,\ldots,A_{hk})
\right).
\]

To determine the entries of the corresponding RT-distance matrix, let
$i,j\in{\rm Im}(f)=\{0,1,\ldots,kh\}$
with
$i\neq j$.
Write
\[
i=mh+l,
\qquad
j=ah+b,
\]
where
\[
0\le l,b\le h-1.
\]
Without loss of generality, assume that
$i>j$.

The proof naturally splits into the following five cases, according to the relative values of $l$ and $b$.

        \begin{center}
            \begin{tabular}{|p{2cm}|p{8cm}|}
            \hline
            \centering\arraybackslash
                Case 1 & \centering\arraybackslash$l=b=0$  \\
            \hline
                 \centering\arraybackslash Case 2 & \centering\arraybackslash$l=0,\,b\neq 0$\\
            \hline
                 \centering\arraybackslash Case 3 & \centering\arraybackslash$l\neq 0,\,b=0$\\
            \hline
                \centering\arraybackslash Case 4 & \centering\arraybackslash$l=b\neq0$ \\
            \hline
                \centering\arraybackslash Case 5 & \centering\arraybackslash$l\neq b,\,l\neq 0,\,b\neq0$\\
            \hline
            \end{tabular}
        \end{center}
        Next, we will analyze the value of $d_{RT}(A_i,A_j)$ in each case.
        
        {\bf Case 1.} If $l=b=0$, then $m>a$, $\big\lceil{\frac{i}{h}}\big\rceil=m$, $\big\lceil{\frac{j}{h}}\big\rceil=a$, $h-(i \text{ smod } h)+1=1$, $h-(j \text{ smod } h)+1=1$, thus $A_i=\begin{pmatrix}
            \overset{m}{\overbrace{1  \cdots 1}} & 0 & 0 &\cdots & 0\\
            0 \cdots 0 & 0 & 0 & \cdots & 0\\
            \vdots \ddots \vdots & \vdots & \vdots& \ddots& \vdots\\
            0\cdots 0 & 0 & 0 & \cdots & 0\\
        \end{pmatrix},\,
        A_j=\begin{pmatrix}
            \overset{a}{\overbrace{1  \cdots 1}} & 0 & 0 &\cdots & 0\\
            0 \cdots 0 & 0 & 0 & \cdots & 0\\
            \vdots \ddots \vdots & \vdots & \vdots& \ddots& \vdots\\
            0\cdots 0 & 0 & 0 & \cdots & 0\\
        \end{pmatrix}$, $d_{RT}(A_i,A_j)=(m-a)h=i-j+(\min\{\,i,j\,\}\mod h)$.

        {\bf Case 2.} If $l=0,\,b\neq 0$, then $m>a$, similarly, $A_i=\begin{pmatrix}
            \overset{m}{\overbrace{1  \cdots 1}} & 0 & 0 &\cdots & 0\\
            0 \cdots 0 & 0 & 0 & \cdots & 0\\
            \vdots \ddots \vdots & \vdots & \vdots& \ddots& \vdots\\
            0\cdots 0 & 0 & 0 & \cdots & 0\\
        \end{pmatrix}$, and ${\big\lceil{\frac{j}{h}}\big\rceil=a+1}$, ${h-(j \text{ smod } h)+1=h-b+1}$, thus 
        $
        A_j=\begin{pmatrix}
            \overset{a}{\overbrace{1  \cdots 1}} & 0 & 0 &\cdots & 0\\
            0 \cdots 0 & 0 & 0 & \cdots & 0\\
            \vdots \ddots \vdots & \vdots & \vdots& \ddots& \vdots\\
            0 \cdots 0 & 1& 0& \cdots & 0\\
            \vdots \ddots \vdots & \vdots & \vdots& \ddots& \vdots\\
            0\cdots 0 & 0 & 0 & \cdots & 0\\
        \end{pmatrix}$. 
        $$ d_{RT}(A_i,A_j)=(m-a)h=i-j+b=i-j+(\min\{\,i,j\,\}\mod h). $$

        {\bf Case 3.} If $l\neq 0,\,b=0$, then $m\geq a$. $\big\lceil{\frac{i}{h}}\big\rceil=m+1$, $h-(i \text{ smod } h)+1=h-l+1$, $A_i=\begin{pmatrix}
            \overset{m}{\overbrace{1  \cdots 1}} & 0 & 0 &\cdots & 0\\
            0 \cdots 0 & 0 & 0 & \cdots & 0\\
            \vdots \ddots \vdots & \vdots & \vdots& \ddots& \vdots\\
            0 \cdots 0 & 1& 0& \cdots & 0\\
            \vdots \ddots \vdots & \vdots & \vdots& \ddots& \vdots\\
            0\cdots 0 & 0 & 0 & \cdots & 0\\
        \end{pmatrix}$.
        In this case, $\big\lceil{\frac{j}{h}}\big\rceil=a$, $h-(j \text{ smod } h)+1=1$, then, $A_j=\begin{pmatrix}
            \overset{a}{\overbrace{1  \cdots 1}} & 0 & 0 &\cdots & 0\\
            0 \cdots 0 & 0 & 0 & \cdots & 0\\
            \vdots \ddots \vdots & \vdots & \vdots& \ddots& \vdots\\
            0\cdots 0 & 0 & 0 & \cdots & 0\\
        \end{pmatrix}$. If $m=a$, $d_{RT}(A_i,A_j)=l=i-j+(\min\{\,i,j\,\}\mod h)$. If $m>a$, $d_{RT}(A_i,A_j)=(m-a)h+l=i-j+b=i-j+(\min\{\,i,j\,\}\mod h)$.

        {\bf Case 4.} If $l=b\neq 0$, then $m>a$. $\big\lceil{\frac{i}{h}}\big\rceil=m+1$, $h-(i \text{ smod } h)+1=h-l+1$, thus, $A_i=\begin{pmatrix}
            \overset{m}{\overbrace{1  \cdots 1}} & 0 & 0 &\cdots & 0\\
            0 \cdots 0 & 0 & 0 & \cdots & 0\\
            \vdots \ddots \vdots & \vdots & \vdots& \ddots& \vdots\\
            0 \cdots 0 & 1& 0& \cdots & 0\\
            \vdots \ddots \vdots & \vdots & \vdots& \ddots& \vdots\\
            0\cdots 0 & 0 & 0 & \cdots & 0\\
        \end{pmatrix}$. ${\big\lceil{\frac{j}{h}}\big\rceil=a+1}$, ${h-(j \text{ smod } h)+1=h-b+1}$, thus 
        $
        A_j=\begin{pmatrix}
            \overset{a}{\overbrace{1  \cdots 1}} & 0 & 0 &\cdots & 0\\
            0 \cdots 0 & 0 & 0 & \cdots & 0\\
            \vdots \ddots \vdots & \vdots & \vdots& \ddots& \vdots\\
            0 \cdots 0 & 1& 0& \cdots & 0\\
            \vdots \ddots \vdots & \vdots & \vdots& \ddots& \vdots\\
            0\cdots 0 & 0 & 0 & \cdots & 0\\
        \end{pmatrix}$. By Definition \ref{def RT-distance}, $${d_{RT}(A_i,A_j)=(m-a)h+l=i-j+b=i-j+(\min\{\,i,j\,\}\mod h}).$$

        {\bf Case 5.} If $l\neq b,\,l\neq0,\,b\neq 0$. $\big\lceil{\frac{i}{h}}\big\rceil=m+1$, $h-(i \text{ smod } h)+1=h-l+1$, $A_i=\begin{pmatrix}
            \overset{m}{\overbrace{1  \cdots 1}} & 0 & 0 &\cdots & 0\\
            0 \cdots 0 & 0 & 0 & \cdots & 0\\
            \vdots \ddots \vdots & \vdots & \vdots& \ddots& \vdots\\
            0 \cdots 0 & 1& 0& \cdots & 0\\
            \vdots \ddots \vdots & \vdots & \vdots& \ddots& \vdots\\
            0\cdots 0 & 0 & 0 & \cdots & 0\\
        \end{pmatrix}$. Similarly,
        $
        A_j=\begin{pmatrix}
            \overset{a}{\overbrace{1  \cdots 1}} & 0 & 0 &\cdots & 0\\
            0 \cdots 0 & 0 & 0 & \cdots & 0\\
            \vdots \ddots \vdots & \vdots & \vdots& \ddots& \vdots\\
            0 \cdots 0 & 1& 0& \cdots & 0\\
            \vdots \ddots \vdots & \vdots & \vdots& \ddots& \vdots\\
            0\cdots 0 & 0 & 0 & \cdots & 0\\
        \end{pmatrix}$. If $m=a$, we have $l>b$, then ${d_{RT}(A_i,A_j)=l=i-j+(\min\{\,i,j\,\}\mod h)}$. If $m>a$, then ${d_{RT}(A_i,A_j)=(m-a)h+l=i-j+b=i-j+(\min\{\,i,j\,\}\mod h)}$.

Consequently, for every
\[
i,j\in\{0,1,\ldots,kh\},
\qquad
i\neq j,
\]
we obtain
\[
d_{RT}(A_i,A_j)
=
|i-j|
+
(\min\{i,j\}\bmod h).
\]

Therefore, by Definition~\ref{Def RT-distance matrix},
\[
[D_{RT}^{f}(h,k,t,A_0,A_1,\ldots,A_{kh})]_{(i+1)(j+1)}
=
\begin{cases}
\left[
2t+1-d_{RT}(A_i,A_j)
\right]^+,
&
i\neq j,
\\[1ex]
0,
&
i=j.
\end{cases}
\]

Substituting the above expression for
$d_{RT}(A_i,A_j)$
gives
\[
[D_{RT}^{f}(h,k,t,A_0,\ldots,A_{kh})]_{(i+1)(j+1)}
=
\begin{cases}
\left[
2t+1-|i-j|-(\min\{i,j\}\bmod h)
\right]^+,
&
i\neq j,
\\[1ex]
0,
&
i=j.
\end{cases}
\]

Hence,
\[
D_{RT}^{f}(h,k,t,A_0,A_1,\ldots,A_{kh})
=
D_1.
\]
Applying Corollary~\ref{Cor1}, we conclude that
\[
r_{RT}^{f}(h,k,t)
\ge
N_{RT}^{h}(D_1).
\]

Combining this inequality with the upper bound established previously yields
\[
N_{RT}^{h}(D_1)
\le
r_{RT}^{f}(h,k,t)
\le
N_{RT}^{h}(D_2),
\]
thereby completing the proof.
\end{proof}

\begin{Remark}
The case $h=1$ deserves particular attention. Since
\[
\min\{i,j\}\bmod h=0,
\qquad
\forall\,i,j\in[E],
\]
we obtain
\[
D_1=D_2,
\]
and consequently
\[
N_{RT}^{h}(D_1)=N_{RT}^{h}(D_2).
\]
Therefore, the lower and upper bounds in Theorem~\ref{Upper-Lower bounds RT-weight} coincide, providing the exact value of the optimal redundancy. Moreover, because the RT metric coincides with the Hamming metric when $h=1$, Theorem~\ref{Upper-Lower bounds RT-weight} recovers exactly Lemma~6 of~\cite{lbwy}.
\end{Remark}

    \begin{Example}
        Let $h=2,\,k=2,\,t=1$. Let $w_{RT}$ be the RT-weight function over $\mathbb{F}_2^{2\times 2}$. Then, the matrices $D_1$ and $D_2$ in Theorem \ref{Upper-Lower bounds RT-weight} are 
        $$ D_1=\begin{pmatrix}
            0 & 2 & 1 & 0 & 0\\
            2 & 0 & 1 & 0 & 0\\
            1 & 1 & 0 & 2 & 1\\
            0 & 0 & 2 & 0 & 1\\
            0 & 0 & 1 & 1 & 0\\
        \end{pmatrix}, ~D_2=\begin{pmatrix}
            0 & 2 & 1 & 0 & 0\\
            2 & 0 & 2 & 1 & 0\\
            1 & 2 & 0 & 2 & 1\\
            0 & 1 & 2 & 0 & 2\\
            0 & 0 & 1 & 2 & 0\\
        \end{pmatrix}. $$ By computation, $N_{RT}^h(D_1)=N_{RT}^h(D_2)=2$, thus, $r_{RT}^{w_{RT}}(2,2,1)=2$. 
    \end{Example}

The bound established in Theorem~\ref{Upper-Lower bounds RT-weight} still involves the quantity
$N_{RT}^{h}(D_1),$ whose exact evaluation may be difficult in general. By combining the explicit structure of the matrix $D_1$ with the lower bound obtained in Lemma~\ref{Lem3}, we derive the following closed-form estimate for the optimal redundancy of FCRTCs associated with the RT-weight function.

\begin{Corollary}
Let
\[
f:\mathbb{F}_2^{h\times k}\longrightarrow{\rm Im}(f)
\]
be the RT-weight function. If
\[
k>\frac{t}{h},
\]
then
\[
r_{RT}^{f}(h,k,t)
\geq
\frac{2^{h+1}}
{(t+2)^2\bigl((h-1)2^h+1\bigr)}
\bigg(
\frac{5t^3+15t^2+10t}{6}
-\phi_1(h,t)
-\phi_2(h,t)
\bigg),
\]
where \[\phi_1(h,t)=\frac{h(h-1)\lfloor\frac{t}{h}\rfloor}{2}
\cdot \frac{(3\lfloor\frac{t}{h}\rfloor-1)h+6(t-h\lfloor\frac{t}{h}\rfloor)+8}{6},\]
and \[\phi_2(h,t)=\frac{(t-h\lfloor\frac{t}{h}\rfloor)(t-h\lfloor\frac{t}{h}\rfloor+1)(t-h\lfloor\frac{t}{h}\rfloor+2)}{6}.\]
\end{Corollary}

  \begin{proof}
By Theorem~\ref{Upper-Lower bounds RT-weight},
\[
r_{RT}^f(h,k,t)\geq N_{RT}^h(D_1),
\]
where the matrix $D_1$ is defined by
\[
[D_1]_{ij}
=
\begin{cases}
\bigl[2t+1-|i-j|-(\min\{i-1,j-1\}\bmod h)\bigr]^+,
& \text{if } i\neq j,\\
0,
& \text{otherwise},
\end{cases}
\]
for $1\leq i,j\leq kh+1$.

Let
\[
\mathcal P=\{P_1,P_2,\ldots,P_{kh+1}\}
\]
be a $D_1$-RT code over
$\mathbb F_2^{h\times N_{RT}^h(D_1)}$.
Without loss of generality, assume that
\[
d_{RT}(P_i,P_j)\geq [D_1]_{ij},
\qquad
1\leq i,j\leq kh+1.
\]

Let $B$ be the submatrix formed  by the first $t+2$ rows and $t+2$ columns of the matrix $D_1$. Consider the subcode
\[
\mathcal P'
=
\{P_1,P_2,\ldots,P_{t+2}\}.
\]
It follows that
\[
d_{RT}(P_i,P_j)
\geq
[D_1]_{ij}
=
[B]_{ij},
\qquad
1\leq i,j\leq t+2.
\]
Hence, by Definition~\ref{Def D-RT code},
$\mathcal P'$ is a $B$-RT code over
$\mathbb F_2^{h\times N_{RT}^h(D_1)}$.
Consequently,
\[
N_{RT}^h(B)
\leq
N_{RT}^h(D_1)
\leq
r_{RT}^f(h,k,t).
\]

Applying Lemma~\ref{Lem3}, we obtain
\[
N_{RT}^h(B)
\geq
\frac{2^{h+1}}
{(t+2)^2\bigl((h-1)2^h+1\bigr)}
\sum_{1\leq i<j\leq t+2}
[B]_{ij}.
\]
Therefore,
\[
r_{RT}^f(h,k,t)
\geq
\frac{2^{h+1}}
{(t+2)^2\bigl((h-1)2^h+1\bigr)}
\sum_{1\leq i<j\leq t+2}
(2t+1-|i-j|-(\min\{i-1,j-1\}\bmod h)\bigr).
\]
Evaluating the above summation yields
\[
\begin{aligned}
r_{RT}^{f}(h,k,t)
&\geq
\frac{2^{h+1}}
{(t+2)^2\bigl((h-1)2^h+1\bigr)}
\bigg(
\frac{5t^3+15t^2+10t}{6}
-\frac{h(h-1)\lfloor\frac{t}{h}\rfloor}{2} \\
&\quad
\cdot \frac{(3\lfloor\frac{t}{h}\rfloor-1)h+6(t-h\lfloor\frac{t}{h}\rfloor)+8}{6}
-\frac{(t-h\lfloor\frac{t}{h}\rfloor)(t-h\lfloor\frac{t}{h}\rfloor+1)(t-h\lfloor\frac{t}{h}\rfloor+2)}{6}
\bigg),
\end{aligned}
\]
which completes the proof.
\end{proof}

\section{FCRTCs for RT-Weight Distribution Functions}

In the previous section, we investigated FCRTCs associated with the RT-weight function and derived explicit bounds on their optimal redundancy. We now consider a broader and more informative class of functions, namely RT-weight distribution functions. Unlike the RT-weight, which records only the total RT weight of a matrix, an RT-weight distribution function captures the complete distribution of RT weights among the columns. Consequently, it preserves significantly more structural information and leads to a richer class of function-correcting RT codes.

The main objective of this section is twofold. First, we construct explicit FCRTCs for RT-weight distribution functions. Second, we determine the optimal redundancy in several important cases, thereby extending the results obtained for the RT-weight function.

Throughout this section, let $T$ be a positive integer, and let $w_{RT}$denote the RT-weight function over
$\mathbb{F}_2^{h\times k}$.

We now introduce the class of RT-weight distribution functions considered in this paper. These functions partition the set of matrices by RT weight intervals rather than by their exact RT weights. Consequently, matrices whose RT weights fall within the same interval are assigned the same function value, yielding a coarser classification than the RT-weight function itself.

 \begin{Definition}\label{def RT-weight distribution}
Let $h$ and $k$ be positive integers, and let $T$ be a positive integer. The \emph{RT-weight distribution function}
$\Delta_{RT}^{h,T}$ over
$\mathbb{F}_2^{h\times k}
$
is defined by
\[
\Delta_{RT}^{h,T}:
\mathbb{F}_2^{h\times k}
\longrightarrow
{\rm Im}(\Delta_{RT}^{h,T}),
\qquad
A
\longmapsto
\left\lfloor
\frac{w_{RT}(A)}{T}
\right\rfloor.
\]
\end{Definition}

We emphasize that the RT-weight distribution function provides a coarser classification than the RT-weight function. Indeed, matrices whose RT weights belong to the same interval of length $T$ are assigned the same function value. Consequently, the parameter $T$ controls the granularity of the partition of $\mathbb{F}_2^{h\times k}$ induced by the function

Since
\[
{\rm Im}(w_{RT})
=
\{0,1,\ldots,kh\},
\]
Definition~\ref{def RT-weight distribution} implies that
\[
{\rm Im}(\Delta_{RT}^{h,T})
=
\left\{
0,1,\ldots,
\left\lfloor\frac{kh}{T}\right\rfloor
\right\}.
\]

To simplify the presentation of the subsequent results, we shall assume throughout this section that
\[
T\mid(kh+1).
\]
Consequently,
\[
\left|{\rm Im}(\Delta_{RT}^{h,T})\right|
=
\frac{kh+1}{T}.
\]
We shall therefore write
\[
E
=
\left|{\rm Im}(\Delta_{RT}^{h,T})\right|
=
\frac{kh+1}{T}.
\]

We now present an explicit construction of function-correcting RT codes for RT-weight distribution functions. The proposed construction exploits the function's particular structure and yields redundancy optimal over a range of parameters.

    \begin{Theorem}\label{Con FCRTC for RT-weight distribution}
        Let $h,k$ be positive integers and $T\geq 2ht+1$. Define an encoding function $$Enc:\mathbb{F}_2^{h\times k}\rightarrow\mathbb{F}_2^{h\times \big(k+\big\lceil{\frac{2t}{h}}\big\rceil\big)}, ~A\mapsto(A,P_{(w_{RT}(A)+1) \text{ smod }T}),$$
        where $P_i\in \mathbb{F}_2^{h\times \big\lceil{\frac{2t}{h}}\big\rceil} $ for $i \in \{\,1,2,\ldots,T\,\}$, and $$P_i=\begin{pmatrix}
            \overset{i-1}{\overbrace{1  \cdots 1}} & \overset{\big\lceil{\frac{2t}{h}}\big\rceil-i+1}{\overbrace{0  \cdots 0}}\\
            0 \cdots 0 & 0 \cdots 0\\
            \vdots \ddots \vdots & \vdots \ddots \vdots\\
            0 \cdots 0 & 0 \cdots 0\\
        \end{pmatrix} \text{ for } i\in \bigg\{\,1,2,\ldots,\bigg\lceil{\frac{2t}{h}}\bigg\rceil+1\bigg\},$$ for an integer $2\leq j\leq h$, if $i\in \big\{\,(j-1)\big\lceil{\frac{2t}{h}}\big\rceil+2,(j-1)\big\lceil{\frac{2t}{h}}\big\rceil+3,\ldots,j\big\lceil{\frac{2t}{h}}\big\rceil+1\big\}$,
        \noindent\begin{minipage}{\linewidth}
        \[
        P_i=\begin{pmatrix}
            \overset{i-(j-1)\big\lceil{\frac{2t}{h}}\big\rceil-1}{\overbrace{0  \cdots 0}} & \overset{j\big\lceil{\frac{2t}{h}}\big\rceil-i+1}{\overbrace{1  \cdots 1}}\\
            \vdots \ddots \vdots & \vdots \ddots \vdots\\
            0 \cdots 0 & 0 \cdots 0\\
            1 \cdots 1 & 0 \cdots 0 \tikzmark{row4}\\
            0 \cdots 0 & 0 \cdots 0\\
            \vdots \ddots \vdots & \vdots \ddots \vdots\\
            0 \cdots 0 & 0 \cdots 0\\
        \end{pmatrix},
        \]
        \begin{tikzpicture}[overlay, remember picture]
            \draw[->, black, thick]
                ([xshift=1.5cm, yshift=0.6ex]pic cs:row4) --
                ([xshift=0.7cm, yshift=0.6ex]pic cs:row4);

            \node[
                anchor=west,
                inner sep=0pt,
                outer sep=0pt,
                xshift=1.7cm,
                yshift=0.6ex,
                text=black
            ] at (pic cs:row4) {the $j$-th row};
        \end{tikzpicture}
        \end{minipage}
        and 
        $$ P_i=\begin{pmatrix}
            \overset{\big\lceil{\frac{2t}{h}}\big\rceil}{\overbrace{1  \cdots 1}}\\
            1\cdots 1\\
            \vdots \ddots\vdots\\
            1 \cdots 1\\
        \end{pmatrix}, \text{ for } ~i \in\bigg\{\,h\bigg\lceil{\frac{2t}{h}}\bigg\rceil+2,h\bigg\lceil{\frac{2t}{h}}\bigg\rceil+3,\ldots,T\,\bigg\}. $$
        Then $\big\{\,Enc(A)\,|\,A\in \mathbb{F}_2^{h\times k}\,\big\}$ is an {\rm FCRTC} for the {\rm RT}-weight distribution function over $\mathbb{F}_2^{h\times k}$. If $T\geq 2ht+1$, $r_{RT}^{\Delta_{RT}^{h,T}}(h,k,t)\leq \big\lceil{\frac{2t}{h}}\big\rceil$. Moreover, if $2ht+1\leq T\leq \frac{kh+1}{2^{(h-1)k}+1}$, $r_{RT}^{\Delta_{RT}^{h,T}}(h,k,t)= \big\lceil{\frac{2t}{h}}\big\rceil$.
    \end{Theorem}

    \begin{proof}
    
    If $h=1$, then the RT-weight function over $\mathbb{F}_2^{h\times k}$ coincides with the Hamming weight function. In this case, whenever $T\geq 2t+1$, the encoding function $Enc$ reduces to
\[
Enc:\mathbb{F}_2^{k}\longrightarrow\mathbb{F}_2^{k+2t},
\qquad
A\longmapsto\bigl(A,P_{(w_{RT}(A)+1)\,\mathrm{smod}\,T}\bigr),
\]
where
\[
P_i=
(\underbrace{1\cdots1}_{i-1}\,
\underbrace{0\cdots0}_{2t-i+1}),
\qquad
1\leq i\leq 2t+1,
\]
and
\[
P_i=
(\underbrace{1\cdots1}_{2t}),
\qquad
2t+2\leq i\leq T.
\]
It follows that
\[
\bigl\{\,Enc(A)\mid A\in\mathbb{F}_2^{k}\,\bigr\}
\]
is an FCRTC. The proof is similar to that of Lemma~8 in~\cite{lbwy}.
Assume that $h\geq 2$. Let $A_1,A_2\in\mathbb{F}_2^{h\times k}$ satisfy
\[
\Delta_{RT}^{h,T}(A_1)\neq\Delta_{RT}^{h,T}(A_2).
\]

If
\[
d_{RT}(A_1,A_2)\geq 2t+1,
\]
then
\[
d_{RT}(Enc(A_1),Enc(A_2))
\geq
d_{RT}(A_1,A_2)
\geq
2t+1.
\]

Now suppose that
\[
d_{RT}(A_1,A_2)\leq 2t.
\]
By the triangle inequality,
\[
|w_{RT}(A_1)-w_{RT}(A_2)|
\leq
d_{RT}(A_1,A_2)
\leq
2t.
\]
Since $T\geq 2ht+1$, it follows that
\[
|w_{RT}(A_1)-w_{RT}(A_2)|<T.
\]

Without loss of generality, write
\[
\Delta_{RT}^{h,T}(A_1)=mT+w_1,
\qquad
\Delta_{RT}^{h,T}(A_2)=(m-1)T+w_2,
\]
for some integer $m$. Then
\[
T+w_1-w_2<T,
\]
which implies that
\[
w_1<w_2.
\]
Moreover,
\[
d_{RT}(A_1,A_2)\geq T+w_1-w_2.
\]

By the definition of the shifted modulo operator,
\[
P_{(w_{RT}(A_1)+1)\,\mathrm{smod}\,T}
=
P_{w_1+1},
\qquad
P_{(w_{RT}(A_2)+1)\,\mathrm{smod}\,T}
=
P_{w_2+1}.
\]
We distinguish the following seven cases according to the values of $w_1$ and $w_2$, and show in each case that
\[
d_{RT}(Enc(A_1),Enc(A_2))
\geq
2t+1.
\]
        {
        \begin{table}[!htbp]
        \begin{center}
            \begin{tabular}{|c|c|c|c|}
             \hline
                & \!\multirow{2}{*}{${0\leq w_1\leq  \big\lceil{\frac{2t}{h}}\big\rceil}$}\! & \!${i\big\lceil{\frac{2t}{h}}\big\rceil+1\leq w_1 \leq (i+1)\big\lceil{\frac{2t}{h}}\big\rceil}$,\!& \!\multirow{2}{*}{${h\big\lceil{\frac{2t}{h}}\big\rceil\leq w_1}$}\!\\ & & ${1\leq i\leq h-1}$ & \\
            \hline
                \!\multirow{2}{*}{${0\leq w_2\leq  \big\lceil{\frac{2t}{h}}\big\rceil}$}\! & \multirow{2}{*}{Case 1}  & \multirow{2}{*}{$\ast$}
                & \multirow{2}{*}{$\ast$} \\
                & & & \\
            \hline
            \!{${j\big\lceil{\frac{2t}{h}}\big\rceil+1\leq w_2 \leq (j+1)\big\lceil{\frac{2t}{h}}\big\rceil},$} \!& \multirow{2}{*}{Case 4}& $i=j$: Case 2 & \multirow{2}{*}{*}\\
            \!${1\leq j\leq h-1}$\! &  & {$i<j$: Case 6} & \\
            \hline
            \!\multirow{2}{*}{${h\big\lceil{\frac{2t}{h}}\big\rceil\leq w_2 }$}\! & \multirow{2}{*}{ Case 5} & \multirow{2}{*}{Case 7} &\multirow{2}{*}{Case 3} \\
            & & & \\
            \hline
            \end{tabular}
        \end{center}
        \end{table}}

        {\bf Case 1.} If $0\leq w_1,w_2\leq \big\lceil{\frac{2t}{h}}\big\rceil$, $P_{w_1+1}=\begin{pmatrix}
            \overset{w_1}{\overbrace{1  \cdots 1}} & \overset{\big\lceil{\frac{2t}{h}}\big\rceil-w_1}{\overbrace{0  \cdots 0}}\\
            0 \cdots 0 & 0 \cdots 0\\
            \vdots \ddots \vdots & \vdots \ddots \vdots\\
            0 \cdots 0 & 0 \cdots 0\\
        \end{pmatrix}$, $P_{w_2+1}=\begin{pmatrix}
            \overset{w_2}{\overbrace{1  \cdots 1}} & \overset{\big\lceil{\frac{2t}{h}}\big\rceil-w_2}{\overbrace{0  \cdots 0}}\\
            0 \cdots 0 & 0 \cdots 0\\
            \vdots \ddots \vdots & \vdots \ddots \vdots\\
            0 \cdots 0 & 0 \cdots 0\\
        \end{pmatrix}$. Then, $d_{RT}(P_{w_1+1},P_{w_2+1})=h(w_2-w_1)$. Thus, $$d_{RT}(Enc(A_1),Enc(A_2))=d_{RT}(A_1,A_2)+d_{RT}(P_{w_1+1},P_{w_2+1})\geq T+w_1-w_2+w_2-w_1.$$
        Since $T\geq2ht+1\geq 2t+1$, $d_{RT}(Enc(A_1),Enc(A_2))\geq T\geq 2t+1$.

        {\bf Case 2.} If $i\big\lceil{\frac{2t}{h}}\big\rceil+1\leq w_1,w_2\leq (i+1)\big\lceil{\frac{2t}{h}}\big\rceil$, $1\leq i \leq h-1$, then 
        
        \noindent\begin{minipage}{\linewidth}
        \[
         P_{w_1+1}=\begin{pmatrix}
            \overset{w_1-i\big\lceil{\frac{2t}{h}}\big\rceil}{\overbrace{0  \cdots 0}} & \overset{(i+1)\big\lceil{\frac{2t}{h}}\big\rceil-w_1}{\overbrace{1  \cdots 1}}\\
            \vdots \ddots \vdots & \vdots \ddots \vdots\\
            0 \cdots 0 & 0 \cdots 0\\
           1 \cdots 1 & 0 \cdots 0 \tikzmark{row1}\\
            0 \cdots 0 & 0 \cdots 0\\
            \vdots \ddots \vdots & \vdots \ddots \vdots\\
            0 \cdots 0 & 0 \cdots 0\\
        \end{pmatrix},
        \]
         \begin{tikzpicture}[overlay, remember picture]
            % 先画箭头，从-1.8cm到-0.2cm，不贯穿文字
            \draw[->, black, thick]
                ([xshift=1.7cm, yshift=0.6ex]pic cs:row1) --
                ([xshift=0.9cm, yshift=0.6ex]pic cs:row1);
            
            % 再单独放置文字，在箭头左侧
            \node[
                anchor=west,
                inner sep=0pt,
                outer sep=0pt,
                xshift=1.7cm,
                yshift=0.6ex,
                text=black
            ] at (pic cs:row1) {the $(i+1)$-th row};
        \end{tikzpicture}
        \end{minipage}
         \noindent\begin{minipage}{\linewidth}
         \[P_{w_2+1}=\begin{pmatrix}
            \overset{w_2-i\big\lceil{\frac{2t}{h}}\big\rceil}{\overbrace{0  \cdots 0}} & \overset{(i+1)\big\lceil{\frac{2t}{h}}\big\rceil-w_2}{\overbrace{1  \cdots 1}}\\
            \vdots \ddots \vdots & \vdots \ddots \vdots\\
            0 \cdots 0 & 0 \cdots 0\\
             1 \cdots 1 & 0 \cdots 0 \tikzmark{row2}\\
            0 \cdots 0 & 0 \cdots 0\\
            \vdots \ddots \vdots & \vdots \ddots \vdots\\
            0 \cdots 0 & 0 \cdots 0\\
        \end{pmatrix}.
        \]
        \begin{tikzpicture}[overlay, remember picture]
            % 先画箭头，从-1.8cm到-0.2cm，不贯穿文字
            \draw[->, black, thick]
                ([xshift=1.7cm, yshift=0.6ex]pic cs:row2) --
                ([xshift=0.9cm, yshift=0.6ex]pic cs:row2);
            
            % 再单独放置文字，在箭头左侧
            \node[
                anchor=west,
                inner sep=0pt,
                outer sep=0pt,
                xshift=1.9cm,
                yshift=0.6ex,
                text=black
            ] at (pic cs:row2) {the $(i+1)$-th row};
        \end{tikzpicture}
         \end{minipage}
         Thus, $ d_{RT}(P_{w_1+1},P_{w_2+1})=h(w_2-w_1) $. Similar to Case 1, $$d_{RT}(Enc(A_1),Enc(A_2))\geq T\geq 2t+1.$$

         {\bf Case 3.} If $w_1,w_2\geq h\big\lceil{\frac{2t}{h}}\big\rceil+1$, then $P_{w_1+1}=P_{w_2+1}=\begin{pmatrix}
            \overset{\big\lceil{\frac{2t}{h}}\big\rceil}{\overbrace{1  \cdots 1}}\\
            1\cdots 1\\
            \vdots \ddots\vdots\\
            1 \cdots 1\\
        \end{pmatrix}$. Since $ {w_1 \geq h\big\lceil{\frac{2t}{h}}\big\rceil+1}$ and $\big\lceil{\frac{2t}{h}}\big\rceil\geq \frac{2t}{h}$, we have $w_1\geq 2t+1$. Thus, $$d_{RT}(Enc(A_1),Enc(A_2))=d_{RT}(A_1,A_2)\geq T-w_2+w_1\geq2t+1.$$

        {\bf Case 4.} If $0\leq w_1\leq \big\lceil{\frac{2t}{h}}\big\rceil$, $i\big\lceil{\frac{2t}{h}}\big\rceil+1\leq w_2\leq (i+1)\big\lceil{\frac{2t}{h}}\big\rceil$, $1\leq i\leq {h-1}$, then 
        \noindent\begin{minipage}{\linewidth}
         $P_{w_1+1}=\begin{pmatrix}
            \overset{w_1}{\overbrace{1  \cdots 1}} & \overset{\big\lceil{\frac{2t}{h}}\big\rceil-w_1}{\overbrace{0  \cdots 0}}\\
            0 \cdots 0 & 0 \cdots 0\\
            \vdots \ddots \vdots & \vdots \ddots \vdots\\
            0 \cdots 0 & 0 \cdots 0\\
        \end{pmatrix},\,
         P_{w_2+1}=\begin{pmatrix}
            \overset{w_2-i\big\lceil{\frac{2t}{h}}\big\rceil}{\overbrace{0  \cdots 0}} & \overset{(i+1)\big\lceil{\frac{2t}{h}}\big\rceil-w_2}{\overbrace{1  \cdots 1}}\\
            \vdots \ddots \vdots & \vdots \ddots \vdots\\
            0 \cdots 0 & 0 \cdots 0\\
             1 \cdots 1 & 0 \cdots 0 \tikzmark{row3}\\
            0 \cdots 0 & 0 \cdots 0\\
            \vdots \ddots \vdots & \vdots \ddots \vdots\\
            0 \cdots 0 & 0 \cdots 0\\
        \end{pmatrix}.
        $
        \begin{tikzpicture}[overlay, remember picture]
            % 先画箭头，从-1.8cm到-0.2cm，不贯穿文字
            \draw[->, black, thick]
                ([xshift=1.7cm, yshift=0.6ex]pic cs:row3) --
                ([xshift=0.9cm, yshift=0.6ex]pic cs:row3);
            
            % 再单独放置文字，在箭头左侧
            \node[
                anchor=west,
                inner sep=0pt,
                outer sep=0pt,
                xshift=1.9cm,
                yshift=0.6ex,
                text=black
            ] at (pic cs:row3) {the $(i+1)$-th row};
        \end{tikzpicture}
         \end{minipage}
      
      Assume that
\[
w_1\geq
w_2-i\left\lceil\frac{2t}{h}\right\rceil.
\]
Then
\[
\left\lceil\frac{2t}{h}\right\rceil-w_1
\leq
(i+1)\left\lceil\frac{2t}{h}\right\rceil-w_2,
\]
and hence
\[
\begin{aligned}
d_{RT}(P_{w_1+1},P_{w_2+1})
&=
h\left(w_2-i\left\lceil\frac{2t}{h}\right\rceil\right)
+h\left(\left\lceil\frac{2t}{h}\right\rceil-w_1\right)\\
&=
h(w_2-w_1)
+h(1-i)\left\lceil\frac{2t}{h}\right\rceil.
\end{aligned}
\]
Therefore,
\[
d_{RT}(Enc(A_1),Enc(A_2))
\geq
T+(h-1)(w_2-w_1)
-h(i-1)\left\lceil\frac{2t}{h}\right\rceil.
\]

Using the inequality
\[
w_2-w_1
\geq
(i-1)\left\lceil\frac{2t}{h}\right\rceil+1,
\]
we obtain
\[
d_{RT}(Enc(A_1),Enc(A_2))
\geq
T-(i-1)\left\lceil\frac{2t}{h}\right\rceil+h-1.
\]
Since $1\leq i\leq h-1$, it follows that
\[
d_{RT}(Enc(A_1),Enc(A_2))
\geq
T-(h-2)\left\lceil\frac{2t}{h}\right\rceil+h-1.
\]
Moreover,
\[
\frac{2t}{h}
\leq
\left\lceil\frac{2t}{h}\right\rceil
<
\frac{2t}{h}+1,
\]
which yields
\[
d_{RT}(Enc(A_1),Enc(A_2))
\geq
(T-2t)+\frac{4t}{h}+1.
\]
Finally, since $T\geq 2ht+1$ and $h\geq 2$, we conclude that
\[
d_{RT}(Enc(A_1),Enc(A_2))
\geq
2t+1.
\]

Assume that
\[
w_1<w_2-i\left\lceil\frac{2t}{h}\right\rceil.
\]
Then
\[
\left\lceil\frac{2t}{h}\right\rceil-w_1
>
(i+1)\left\lceil\frac{2t}{h}\right\rceil-w_2,
\]
and hence
\[
\begin{aligned}
d_{RT}(P_{w_1+1},P_{w_2+1})
&=
hw_1
+(h-i)\left(w_2-i\left\lceil\frac{2t}{h}\right\rceil-w_1\right) \\
&\quad
+h\left((i+1)\left\lceil\frac{2t}{h}\right\rceil-w_2\right)\\
&=
i(w_1-w_2)
+\left(h+i^2\right)\left\lceil\frac{2t}{h}\right\rceil.
\end{aligned}
\]
Therefore,
\[
d_{RT}(Enc(A_1),Enc(A_2))
\geq
T+(i+1)(w_1-w_2)
+\left(h+i^2\right)\left\lceil\frac{2t}{h}\right\rceil.
\]

Using the inequality
\[
w_1-w_2
\geq
-(i+1)\left\lceil\frac{2t}{h}\right\rceil,
\]
we obtain
\[
d_{RT}(Enc(A_1),Enc(A_2))
\geq
T+\left(h-2i-1\right)\left\lceil\frac{2t}{h}\right\rceil.
\]
Since $1\leq i\leq h-1$, we have
\[
h-2i-1\geq 1-h,
\]
and hence
\[
d_{RT}(Enc(A_1),Enc(A_2))
\geq
T-(h-1)\left\lceil\frac{2t}{h}\right\rceil.
\]
Moreover,
\[
\left\lceil\frac{2t}{h}\right\rceil
<
\frac{2t}{h}+1,
\]
which yields
\[
d_{RT}(Enc(A_1),Enc(A_2))
>
T-2t-h+\frac{2t}{h}+1.
\]
Since $T\geq 2ht+1$, it follows that
\[
d_{RT}(Enc(A_1),Enc(A_2))
>
2(h-1)t-h+\frac{2t}{h}+2.
\]

We now distinguish three cases.

- If $h=2$, then
\[
d_{RT}(Enc(A_1),Enc(A_2))
\geq
2t+1.
\]

- If $h=3$, then
\[
d_{RT}(Enc(A_1),Enc(A_2))
>
4t+\frac{2t}{3}-1.
\]
Since $t\geq 1$, we obtain
\[
d_{RT}(Enc(A_1),Enc(A_2))
\geq
2t+1.
\]

- Finally, assume that $h\geq 4$. Since $t\geq 1$,
\[
\begin{aligned}
d_{RT}(Enc(A_1),Enc(A_2))
&>
2(h-1)t-h+\frac{2t}{h}+2\\
&\geq
2(h-1)t-ht+\frac{2t}{h}+2\\
&=
(h-2)t+\frac{2t}{h}+2\\
&\geq
2t+1.
\end{aligned}
\]

        {\bf Case 5.} If $0\leq w_1\leq \big\lceil{\frac{2t}{h}}\big\rceil$, $w_2\geq h\big\lceil{\frac{2t}{h}}\big\rceil+1$, then $P_{w_1+1}=\begin{pmatrix}
            \overset{w_1}{\overbrace{1  \cdots 1}} & \overset{\big\lceil{\frac{2t}{h}}\big\rceil-w_1}{\overbrace{0  \cdots 0}}\\
            0 \cdots 0 & 0 \cdots 0\\
            \vdots \ddots \vdots & \vdots \ddots \vdots\\
            0 \cdots 0 & 0 \cdots 0\\
        \end{pmatrix},$ $P_{w_2+1}=\begin{pmatrix}
            \overset{\big\lceil{\frac{2t}{h}}\big\rceil}{\overbrace{1  \cdots 1}}\\
            1\cdots 1\\
            \vdots \ddots\vdots\\
            1 \cdots 1\\
        \end{pmatrix}$. $d_{RT}(P_{w_1+1},P_{w_2+1})=(h-1)w_1+h\big(\big\lceil{\frac{2t}{h}}\big\rceil-w_1\big)=h\big\lceil{\frac{2t}{h}}\big\rceil-w_1$. Thus, $$d_{RT}(Enc(A_1),Enc(A_2))\geq T-w_2+h\bigg\lceil{\frac{2t}{h}}\bigg\rceil\geq 2t+1.$$

        {\bf Case 6.} If $i\big\lceil{\frac{2t}{h}}\big\rceil+1\leq w_1\leq (i+1)\big\lceil{\frac{2t}{h}}\big\rceil$, $j\big\lceil{\frac{2t}{h}}\big\rceil\leq w_2\leq (j+1)\big\lceil{\frac{2t}{h}}\big\rceil$, $1\leq i<h-1$, $i<j\leq h-1$, then 
        
        \noindent\begin{minipage}{\linewidth}
        \[
         P_{w_1+1}=\begin{pmatrix}
            \overset{w_1-i\big\lceil{\frac{2t}{h}}\big\rceil}{\overbrace{0  \cdots 0}} & \overset{(i+1)\big\lceil{\frac{2t}{h}}\big\rceil-w_1}{\overbrace{1  \cdots 1}}\\
            \vdots \ddots \vdots & \vdots \ddots \vdots\\
            0 \cdots 0 & 0 \cdots 0\\
           1 \cdots 1 & 0 \cdots 0 \tikzmark{row5} \\
            0 \cdots 0 & 0 \cdots 0\\
            \vdots \ddots \vdots & \vdots \ddots \vdots\\
            0 \cdots 0 & 0 \cdots 0\\
        \end{pmatrix},
        \]
         \begin{tikzpicture}[overlay, remember picture]
            % 先画箭头，从-1.8cm到-0.2cm，不贯穿文字
            \draw[->, black, thick]
                ([xshift=1.7cm, yshift=0.6ex]pic cs:row5) --
                ([xshift=0.9cm, yshift=0.6ex]pic cs:row5);
            
            % 再单独放置文字，在箭头左侧
            \node[
                anchor=west,
                inner sep=0pt,
                outer sep=0pt,
                xshift=1.9cm,
                yshift=0.6ex,
                text=black
            ] at (pic cs:row5) {the $(i+1)$-th row};
        \end{tikzpicture}
        \end{minipage}
        \noindent\begin{minipage}{\linewidth}
        \[
         P_{w_2+1}=\begin{pmatrix}
            \overset{w_2-j\big\lceil{\frac{2t}{h}}\big\rceil}{\overbrace{0  \cdots 0}} & \overset{(j+1)\big\lceil{\frac{2t}{h}}\big\rceil-w_2}{\overbrace{1  \cdots 1}}\\
            \vdots \ddots \vdots & \vdots \ddots \vdots\\
            0 \cdots 0 & 0 \cdots 0\\
            1 \cdots 1 & 0 \cdots 0 \tikzmark{row6}\\
            0 \cdots 0 & 0 \cdots 0\\
            \vdots \ddots \vdots & \vdots \ddots \vdots\\
            0 \cdots 0 & 0 \cdots 0\\
        \end{pmatrix}.
        \]
         \begin{tikzpicture}[overlay, remember picture]
            % 先画箭头，从-1.8cm到-0.2cm，不贯穿文字
            \draw[->, black, thick]
                ([xshift=1.7cm, yshift=0.6ex]pic cs:row6) --
                ([xshift=0.9cm, yshift=0.6ex]pic cs:row6);
            
            % 再单独放置文字，在箭头左侧
            \node[
                anchor=west,
                inner sep=0pt,
                outer sep=0pt,
                xshift=1.9cm,
                yshift=0.6ex,
                text=black
            ] at (pic cs:row6) {the $(j+1)$-th row};
        \end{tikzpicture}
        \end{minipage}
      
      Assume that
\[
w_1-i\left\lceil\frac{2t}{h}\right\rceil
\geq
w_2-j\left\lceil\frac{2t}{h}\right\rceil.
\]
Then
\[
(i+1)\left\lceil\frac{2t}{h}\right\rceil-w_1
\leq
(j+1)\left\lceil\frac{2t}{h}\right\rceil-w_2,
\]
and hence
\[
\begin{aligned}
d_{RT}(P_{w_1+1},P_{w_2+1})
&=(h-i)\left(w_2-j\left\lceil\frac{2t}{h}\right\rceil\right)\\
&\quad
+h\left((j-i)\left\lceil\frac{2t}{h}\right\rceil+w_1-w_2\right)\\
&=hw_1-iw_2+i(j-h)\left\lceil\frac{2t}{h}\right\rceil.
\end{aligned}
\]
Therefore,
\[
\begin{aligned}
d_{RT}(Enc(A_1),Enc(A_2))
&\geq
T+(h+1)w_1-(i+1)w_2\\
&\quad
+i(j-h)\left\lceil\frac{2t}{h}\right\rceil.
\end{aligned}
\]

Using the inequalities
\[
w_1\geq i\left\lceil\frac{2t}{h}\right\rceil+1,
\qquad
w_2\leq (j+1)\left\lceil\frac{2t}{h}\right\rceil,
\]
we obtain
\[
\begin{aligned}
d_{RT}(Enc(A_1),Enc(A_2))
&\geq
T+(h+1)\left(i\left\lceil\frac{2t}{h}\right\rceil+1\right)\\
&\quad
-(i+1)(j+1)\left\lceil\frac{2t}{h}\right\rceil
+i(j-h)\left\lceil\frac{2t}{h}\right\rceil\\
&=
T-(j+1)\left\lceil\frac{2t}{h}\right\rceil+h+1.
\end{aligned}
\]
Finally, since $i<j\leq h-1$ and
\[
\left\lceil\frac{2t}{h}\right\rceil
<
\frac{2t}{h}+1,
\]
it follows that
\[
d_{RT}(Enc(A_1),Enc(A_2))
\geq
T-2t+1
\geq
2t+1.
\]

   Assume that
\[
w_1-i\left\lceil\frac{2t}{h}\right\rceil
<
w_2-j\left\lceil\frac{2t}{h}\right\rceil.
\]
Then
\[
(i+1)\left\lceil\frac{2t}{h}\right\rceil-w_1
>
(j+1)\left\lceil\frac{2t}{h}\right\rceil-w_2,
\]
and hence
\[
\begin{aligned}
d_{RT}(P_{w_1+1},P_{w_2+1})
&=(h-i)\left(w_1-i\left\lceil\frac{2t}{h}\right\rceil\right)\\
&\quad
+h\left((i-j)\left\lceil\frac{2t}{h}\right\rceil+w_2-w_1\right)\\
&=hw_2-iw_1+\left(i^2-hj\right)
\left\lceil\frac{2t}{h}\right\rceil.
\end{aligned}
\]
Therefore,
\[
\begin{aligned}
d_{RT}(Enc(A_1),Enc(A_2))
&\geq
T+(h-1)w_2-(i-1)w_1\\
&\quad
+\left(i^2-hj\right)
\left\lceil\frac{2t}{h}\right\rceil.
\end{aligned}
\]

Using the inequalities
\[
w_1\leq (i+1)\left\lceil\frac{2t}{h}\right\rceil,
\qquad
w_2\geq j\left\lceil\frac{2t}{h}\right\rceil+1,
\]
we obtain
\[
d_{RT}(Enc(A_1),Enc(A_2))
\geq
T-(j-1)\left\lceil\frac{2t}{h}\right\rceil+h-1.
\]
Finally, since $i<j\leq h-1$, $T\geq 2ht+1$, and $h\geq 2$, it follows that
\[
d_{RT}(Enc(A_1),Enc(A_2))
\geq
T-2t+1+\frac{4t}{h}
\geq
2t+1.
\]

        {\bf Case 7.} If $i\big\lceil{\frac{2t}{h}}\big\rceil+1\leq w_1\leq (i+1)\big\lceil{\frac{2t}{h}}\big\rceil$, $1\leq i\leq h-1$, $w_2\geq h\big\lceil{\frac{2t}{h}}\big\rceil+1$, then 

        \noindent\begin{minipage}{\linewidth}
        $ P_{w_2+1}=\begin{pmatrix}
            \overset{\big\lceil{\frac{2t}{h}}\big\rceil}{\overbrace{1  \cdots 1}}\\
            1\cdots 1\\
            \vdots \ddots\vdots\\
            1 \cdots 1\\
        \end{pmatrix},\qquad
         P_{w_1+1}=\begin{pmatrix}
            \overset{w_1-i\big\lceil{\frac{2t}{h}}\big\rceil}{\overbrace{0  \cdots 0}} & \overset{(i+1)\big\lceil{\frac{2t}{h}}\big\rceil-w_1}{\overbrace{1  \cdots 1}}\\
            \vdots \ddots \vdots & \vdots \ddots \vdots\\
            0 \cdots 0 & 0 \cdots 0\\
           1 \cdots 1 & 0 \cdots 0 \tikzmark{row7}\\
            0 \cdots 0 & 0 \cdots 0\\
            \vdots \ddots \vdots & \vdots \ddots \vdots\\
            0 \cdots 0 & 0 \cdots 0\\
        \end{pmatrix}.
        $
         \begin{tikzpicture}[overlay, remember picture]
            % 先画箭头，从-1.8cm到-0.2cm，不贯穿文字
            \draw[->, black, thick]
                ([xshift=1.7cm, yshift=0.6ex]pic cs:row7) --
                ([xshift=0.9cm, yshift=0.6ex]pic cs:row7);
            
            % 再单独放置文字，在箭头左侧
            \node[
                anchor=west,
                inner sep=0pt,
                outer sep=0pt,
                xshift=1.9cm,
                yshift=0.6ex,
                text=black
            ] at (pic cs:row7) {the $(i+1)$-th row};
        \end{tikzpicture}
        \end{minipage}
        $d_{RT}(P_{w_1+1},P_{w_2+1})=w_1+(h-i-1)\big\lceil{\frac{2t}{h}}\big\rceil$. Since $d_{RT}(A_1,A_2)\geq T+w_1-w_2$,
        $$d_{RT}(Enc(A_1),Enc(A_2))\geq T-w_2+2w_1+(h-i-1)\bigg\lceil{\frac{2t}{h}}\bigg\rceil.$$

 By the assumption
\[
w_1\geq i\left\lceil\frac{2t}{h}\right\rceil+1,
\]
we obtain
\[
d_{RT}(Enc(A_1),Enc(A_2))
\geq
T-w_2+(h+i-1)\left\lceil\frac{2t}{h}\right\rceil+2.
\]
Since $i\geq 1$, we have $h+i-1\geq h$. Therefore,
\[
d_{RT}(Enc(A_1),Enc(A_2))
\geq
h\left\lceil\frac{2t}{h}\right\rceil+1
\geq
2t+1.
\]

Hence, by Definition~\ref{Def FCRTCs}, if $T\geq 2ht+1$, the set
\[
\{\,Enc(A)\mid A\in\mathbb{F}_2^{h\times k}\,\}
\]
forms an FCRTC for the RT-weight distribution function
$\Delta_{RT}^{h,T}$ over $\mathbb{F}_2^{h\times k}$. Consequently,
\[
r_{RT}^{\Delta_{RT}^{h,T}}(h,k,t)
\leq
\left\lceil\frac{2t}{h}\right\rceil.
\]

Furthermore, if
\[
T\leq \frac{kh+1}{2^{(h-1)k}+1},
\]
then
\[
\big|{\rm Im}(\Delta_{RT}^{h,T})\big|
=
\frac{kh+1}{T}
\geq
2^{(h-1)k}+1.
\]
It follows from Corollary~\ref{Cor1} that
\[
r_{RT}^{\Delta_{RT}^{h,T}}(h,k,t)
\geq
\left\lceil\frac{2t}{h}\right\rceil.
\]
Combining this lower bound with the above upper bound yields
\[
r_{RT}^{\Delta_{RT}^{h,T}}(h,k,t)
=
\left\lceil\frac{2t}{h}\right\rceil,
\]
provided that
\[
2ht+1\leq T\leq
\frac{kh+1}{2^{(h-1)k}+1}.
\]
    \end{proof}

\begin{Remark}
The proof of Lemma~8 in~\cite{lbwy} is based on Corollary~1 of~\cite{lbwy}. Therefore, the additional assumption
$T\leq \frac{k+1}{2}$ is implicitly required for the argument to be valid. With this assumption, the proof given in~\cite{lbwy} applies without modification. Moreover, when $h=1$, the Rosenbloom--Tsfasman metric reduces to the Hamming metric. Hence, Theorem~\ref{Con FCRTC for RT-weight distribution} recovers exactly the result of Lemma~8 in~\cite{lbwy}.
\end{Remark}

\section{FCRTCs for RT-Locally-Two-Valued Binary Functions}

    In this section, we introduce a class of functions called RT-locally-two-valued binary functions. For RT-locally-two-valued binary functions, we mainly study FCRTCs for these functions. In some cases, we obtain the optimal redundancy of these FCRTCs.

\begin{Definition}
Let $h$ and $k$ be positive integers, and let
\[
f:\mathbb{F}_2^{h\times k}\rightarrow{\rm Im}(f)
\]
be a function. For
\[
A\in\mathbb{F}_2^{h\times k}
\]
and
\[
\rho\ge0,
\]
the \emph{RT-function ball} centered at $A$ with radius $\rho$ is defined by
\[
B_f^{RT}(h,A,\rho)
=
\left\{
f(B)
\;\middle|\;
B\in\mathbb{F}_2^{h\times k},
\;
d_{RT}(A,B)\le\rho
\right\}.
\]
\end{Definition}

Based on the notion of RT-function balls introduced above, we are now in a position to define RT-locally-two-valued binary functions. Intuitively, these are functions whose values exhibit only limited local variation with respect to the RT metric: within every RT ball of a prescribed radius, the function takes at most two distinct values.

  \begin{Definition}
Let
$
f:\mathbb{F}_2^{h\times k}\rightarrow{\rm Im}(f)
$
be an arbitrary function, and let
$
\rho
$
be a nonnegative real number. If
$
\left|B_f^{RT}(h,A,\rho)\right|\le2,
\,
\forall\,A\in\mathbb{F}_2^{h\times k},
$
then $f$ is called a \emph{$\rho$-RT-locally-two-valued binary function}.
\end{Definition}

 The class of $\rho$-RT-locally-two-valued binary functions naturally includes all two-valued binary functions. Indeed, since such functions take at most two distinct values over their entire domain, they satisfy the defining condition for every nonnegative real number $\rho$.

Another important example is provided by the RT-weight distribution function. If
\[
T\geq 4t+1,
\]
then
$
\Delta_{RT}^{h,T}
$
is a
$
2t\text{-RT-locally-two-valued binary function}.
$

\begin{Theorem}\label{Th:FCRTCs for locally binary function}
Let $h$ and $k$ be positive integers, and let
\[
f:\mathbb{F}_2^{h\times k}\longrightarrow{\rm Im}(f)
\]
be a $2t$-RT-locally-two-valued binary function. If
\[
|{\rm Im}(f)|
\ge
2^{(h-1)k}+1,
\]
then
\[
r_{RT}^{f}(h,k,t)
=
\left\lceil\frac{2t}{h}\right\rceil.
\]
\end{Theorem}

    \begin{proof}
        Let $f$ be a $2t$-RT-locally-two-valued binary function over $\mathbb{F}_2^{h\times k}$. Define 
        \begin{equation*}
            {\boldsymbol{\alpha}}_{2t}(A)=\left\{\begin{aligned}
                &(\,\overset{\big\lceil{\frac{2t}{h}}\big\rceil}{\overbrace{1\,0\,\cdots\,0}}\,)^{\intercal}, ~\text{if } f(A)=\max{B_f^{RT}(h,A,2t)},\\
                &(\,\overset{\big\lceil{\frac{2t}{h}}\big\rceil}{\overbrace{0\,0\,\cdots\,0}}\,)^{\intercal}, ~\text{otherwise}.
            \end{aligned}
            \right.
        \end{equation*}
        Let the matrix $P(A)=(\,\overset{\big\lceil{\frac{2t}{h}}\big\rceil}{\overbrace{{\boldsymbol{\alpha}}_{2t}(A)~{\boldsymbol{\alpha}}_{2t}(A)~\cdots~{\boldsymbol{\alpha}}_{2t}(A)}}\,)$. 
        
        Define the encoding function
\[
Enc:\mathbb{F}_2^{h\times k}\longrightarrow
\mathbb{F}_2^{h\times \left(k+\left\lceil\frac{2t}{h}\right\rceil\right)},
\qquad
A\longmapsto (A,P(A)),
\]
for every $A\in\mathbb{F}_2^{h\times k}$.

Let $A_1,A_2\in\mathbb{F}_2^{h\times k}$ satisfy
$f(A_1)\neq f(A_2)$.

If $d_{RT}(A_1,A_2)\geq 2t+1$, then
\[
d_{RT}(Enc(A_1),Enc(A_2))
\geq d_{RT}(A_1,A_2)
\geq 2t+1.
\]

Now assume that $d_{RT}(A_1,A_2)\leq 2t$. Then
\[
f(A_1),f(A_2)\in
B_f^{RT}(h,A_1,2t)\cap
B_f^{RT}(h,A_2,2t).
\]
Since $f$ is a $2t$-RT-locally-two-valued binary function, it follows that
\[
B_f^{RT}(h,A_1,2t)=
B_f^{RT}(h,A_2,2t).
\]
Because $f(A_1)\neq f(A_2)$, we must have
\[
{\boldsymbol{\alpha}}_{2t}(A_1)\neq
{\boldsymbol{\alpha}}_{2t}(A_2).
\]
By Definition~\ref{def RT-distance},
\[
d_{RT}\bigl({\boldsymbol{\alpha}}_{2t}(A_1),
{\boldsymbol{\alpha}}_{2t}(A_2)\bigr)=h.
\]
Therefore,
\[
\begin{aligned}
d_{RT}(Enc(A_1),Enc(A_2))
&=d_{RT}(A_1,A_2)
+d_{RT}(P(A_1),P(A_2))\\
&\geq 1+h\left\lceil\frac{2t}{h}\right\rceil
\geq 2t+1.
\end{aligned}
\]
    
    By the definition of FCRTCs, the set $\{\,\mathrm{Enc}(A)\mid A\in\mathbb{F}_2^{h\times k}\,\}
$forms an FCRTC for the function $f$. Consequently,
$r_{RT}^f(h,k,t)\leq\left\lceil\frac{2t}{h}\right\rceil.
$
On the other hand, Corollary~\ref{Cor1} implies that if
$
\big|{\rm Im}(f)\big|\geq 2^{(h-1)k}+1,
$
then
$
r_{RT}^f(h,k,t)\geq\left\lceil\frac{2t}{h}\right\rceil.
$
Combining these upper and lower bounds, we conclude that whenever
$
\big|{\rm Im}(f)\big|\geq 2^{(h-1)k}+1,
$
it holds that
$
r_{RT}^f(h,k,t)=\left\lceil\frac{2t}{h}\right\rceil.
$
        
    \end{proof}

    \begin{Remark}
        If $h=1$, the result of Theorem \ref{Th:FCRTCs for locally binary function} is the same as Lemma 5 in {\rm \cite{lbwy}}.
    \end{Remark}

\section{Conclusion and Future Avenues} 
In this article,  we introduced function-correcting-Rosenbloom--Tsfasman codes (FCRTCs) over $\mathbb{F}_2^{h\times k}$, thereby extending the framework of function-correcting codes from the Hamming metric to the more general Rosenbloom--Tsfasman (RT) metric. Since the RT metric coincides with the Hamming metric when $h=1$, the proposed framework naturally generalizes the function-correcting codes introduced in~\cite{lbwy}. Similar to classical function-correcting codes, FCRTCs reduce the redundancy required to protect function values compared with conventional error-correcting codes. A central contribution of this work is establishing a fundamental connection between the optimal redundancy of FCRTCs and the minimum column length of irregular RT-distance codes. By introducing RT-distance matrices and RT-function-distance matrices, we transformed the problem of determining the optimal redundancy of FCRTCs into the study of irregular RT-distance codes. This approach enabled us to derive general upper and lower bounds on the optimal redundancy for arbitrary functions, together with corresponding bounds on the smallest column length of irregular RT-distance codes associated with symmetric matrices. By exploiting the algebraic structure of specific classes of functions, we obtained sharper results for several important families, including RT-weight functions, RT-weight distribution functions, and RT-locally-two-valued binary functions. In particular, we proposed explicit constructions of FCRTCs for RT-weight distribution functions and showed that these constructions achieve the optimal redundancy for a broad range of parameters. Moreover, when $h=1$, all our results reduce to their counterparts in~\cite{lbwy}, demonstrating that the present work provides a genuine extension of the existing theory from the Hamming metric to the RT metric. 

Several interesting research directions naturally arise from this work. First, it would be worthwhile to investigate FCRTCs across larger classes of functions and determine tighter redundancy bounds for these functions. Second, developing new constructions that attain or improve the theoretical bounds established in this paper remains an important problem. Since FCRTCs employ systematic encoding, their encoding and decoding procedures are inherently determined by the underlying function. Consequently, the design of efficient encoding and decoding algorithms achieving optimal redundancy constitutes another challenging and promising direction for future research. Finally, it would also be interesting to investigate extensions of the proposed framework to other metrics, larger finite fields, and more general algebraic structures, as well as potential applications to communication systems over parallel channels.

\bigskip

\noindent\textbf{Acknowledgement}. This work was supported by the National Natural Science Foundation of China (NSFC) under Grant Nos. 12441102 and 12271199, and by the Fundamental Research Funds for the Central Universities (Grant No. CCNU25JCPT031). S.~Mesnager sincerely thanks Hongwei Liu for many fruitful discussions, for providing an inspiring scientific environment at Central China Normal University in Wuhan, and for the warm hospitality she received during her visit at the end of 2025.

\end{document}